\documentclass[11pt]{article}

\usepackage[letterpaper, margin=1in]{geometry}

\usepackage{setspace}
\usepackage{latexsym,amsthm,amsmath,amssymb}
\usepackage{color}
\usepackage{mathrsfs}
\usepackage{hyperref}
\usepackage{url}
\usepackage{pstricks}
\usepackage{soul}
\usepackage{caption}
\usepackage{subcaption}
\usepackage{tikz,ifthen,calc}
\usepackage{stmaryrd}
\usepackage{mdwlist}
\usepackage{paralist}
\usepackage{xspace}
\usepackage{cleveref}
\usepackage{float}
\usepackage{times}

\usetikzlibrary{positioning, fit}
\usetikzlibrary {arrows.meta}
\usetikzlibrary{positioning}
\usetikzlibrary{shapes}
\usetikzlibrary{shapes.symbols,patterns}
\usetikzlibrary{calc,through,backgrounds}
\usetikzlibrary{decorations.pathreplacing,calligraphy}
\usetikzlibrary{plotmarks}

\makeatletter
\newcommand{\xMapsto}[2][]{\ext@arrow 0599{\Mapstofill@}{#1}{#2}}
\def\Mapstofill@{\arrowfill@{\Mapstochar\Relbar}\Relbar\Rightarrow}
\makeatother

\newtheorem{lemma}{Lemma}[section]
\newtheorem{theorem}[lemma]{Theorem}
\newtheorem{corollary}[lemma]{Corollary}
\newtheorem{observation}[lemma]{Observation}

\newtheorem{claim}[lemma]{Claim}

\newcommand{\lowhourglass}{\mathbin{\rotatebox[origin=c]{-90}{$\bowtie$}}}
\newcommand{\hourglass}{\mathrel{\raisebox{1pt}{$\lowhourglass$}}}

\newcommand{\set}[1]{\left\{ #1 \right\}}

\def\FLOW{\mbox{\sf flow}}

\def\Ghalf{G^{1/2}}
\def\Vhalf{V^{1/2}}
\def\Ehalf{E^{1/2}}

\def\calC{\mathscr{C}}

\newcommand{\calE}{\mathscr{E}}
\newcommand{\calF}{\mathscr{F}}
\newcommand{\calS}{\mathscr{S}}

\newcommand{\DRlong}{\textsc{Degree Realization}\xspace}
\newcommand{\DR}{\textsc{DR}\xspace}

\newcommand{\cliquelong}{\textsc{Maximum Clique}\xspace}
\newcommand{\clique}{\textsc{MC}\xspace}
\newcommand{\cliqueDRlong}{\textsc{Maximum Clique Degree Realization}\xspace}
\newcommand{\cliqueDR}{\textsc{MC-DR}\xspace}

\newcommand{\islong}{\textsc{Maximum Independent Set}\xspace}
\newcommand{\is}{\textsc{MIS}\xspace}
\newcommand{\isDRlong}{\textsc{MIS Degree Realization}\xspace}
\newcommand{\isDR}{\textsc{MIS-DR}\xspace}

\newcommand{\vclong}{\textsc{Minimum Vertex Cover}\xspace}
\newcommand{\vc}{\textsc{MVC}\xspace}
\newcommand{\vcDRlong}{\textsc{MVC Degree Realization}\xspace}
\newcommand{\vcDR}{\textsc{MVC-DR}\xspace}

\newcommand{\matchlong}{\textsc{Maximum Matching}\xspace}
\newcommand{\match}{\textsc{MM}\xspace}
\newcommand{\matchDRlong}{\textsc{MM Degree Realization}\xspace}
\newcommand{\matchDR}{\textsc{MM-DR}\xspace}

\newcommand{\domsetlong}{\textsc{Minimum Dominating Set}\xspace}
\newcommand{\domset}{\textsc{MDS}\xspace}
\newcommand{\domsetDRlong}{\textsc{MDS Degree Realization}\xspace}
\newcommand{\domsetDR}{\textsc{MDS-DR}\xspace}

\newcommand{\InvMatch}{\textbf{Invert}\xspace}

\def\inline#1:{\par\vskip 3pt\noindent{\bf #1:}\hskip 10pt}
\def\midinline#1:{\par\noindent{\bf #1:}\hskip 10pt}
\def\dnsinline#1:{\par\vskip -7pt\noindent{\bf #1:}\hskip 10pt}

\title{\textbf{Efficient 
Optimized Degree Realization: \\
Minimum Dominating Set \& Maximum Matching\footnote{The results have appeared earlier in ``Igor Kalinichev, Efficient Optimized Degree Realization:
Minimum Dominating Set \& Maximum Matching, M.Sc. Thesis, The Weizmann Institute of Science, August 2025."}
}
}

\author{Amotz Bar-Noy%
\thanks{City University of New York (CUNY), USA.
Email: amotz@sci.brooklyn.cuny.edu}
\and
Igor Kalinichev%
\thanks{The Weizmann Institute of Science, Rehovot, Israel.
Email:  igor.kalinichev@weizmann.ac.il}
\and
David Peleg%
\thanks{The Weizmann Institute of Science, Rehovot, Israel.
Email: david.peleg@weizmann.ac.il}
\and
Dror Rawitz%
\thanks{Bar Ilan University, Ramt-Gan, Israel. Email: dror.rawitz@biu.ac.il}
}

\begin{document}

\begin{titlepage}
    
\maketitle

\begin{abstract}
The \DRlong problem requires, given a sequence $d$ of $n$ positive integers, 
to decide whether there exists a graph whose degrees correspond to $d$, 
and to construct such a graph if it exists.
A more challenging variant of the problem arises when $d$ has many different realizations, 
and some of them may be more desirable than others.
We study \emph{optimized realization} problems in which the goal is to compute a realization that optimizes some quality measure.
Efficient algorithms are known for the problems of finding a realization with the maximum clique, the maximum independent set, or the minimum vertex cover. 
In this paper, we focus on two problems for which such algorithms were not known. The first is the \DRlong with \domsetlong problem, where the goal is to find a realization whose minimum dominating set is minimized among all the realizations of the given sequence $d$.
The second is the \DRlong with \matchlong problem, where the goal is to find a realization 
with the largest matching among all the realizations of $d$.
We present polynomial time realization algorithms for these two open problems.

A related problem of interest and importance is \emph{characterizing} the sequences 
with a given value of the optimized function. 
This leads to an efficient computation of the optimized value
without providing the realization that achieves that value.
For the \matchlong problem, a succinct characterization of degree 
sequences with a maximum matching of a given size was known. 
This paper provides a succinct characterization of sequences with minimum dominating set of a given size.
\end{abstract}

\def\thepage{}

\end{titlepage}

\pagenumbering{arabic}

\newpage


\section{Introduction}


Given a non-increasing sequence $d=(d_1,\ldots,d_n)$ of positive integers, 
the \DRlong (\DR) problem requires to decide if $d$ is the degree sequence 
of some $n$-vertex (simple undirected) graph $G=(V,E)$, 
that is, $\deg_G(i)=d_i$ for every $i \in [1,n]$, and to construct such a graph if it exists. 
In such a case, we say that $d$ is \emph{graphic}.
For instance, the sequence $d = (4,3,2,1,1)$ cannot be realized since $\sum_i d_i$ is odd,
and $d' = (4,3,1,1,1)$ cannot be realized despite the fact that $\sum_i d'_i$ is even.
In contrast, $d'' = (4,3,2,2,1)$ is graphic.

The two key questions studied extensively in the past concern identifying characterizations 
(or, necessary and sufficient conditions) for a sequence to be graphic, and developing 
effective and efficient algorithms for finding a realizing graph for a given sequence if it exists.
A necessary and sufficient condition for a given sequence of integers to be
graphic (also implying an 
$O(n)$ decision algorithm) was presented by Erd\H{o}s and Gallai in~\cite{EG60}.
(For alternative proofs see~\cite{AT94,choudum86,DF05,TT08,TVW10,TV03}.)
Havel~\cite{Havel55} and Hakimi~\cite{Hakimi62} 
described an $O(\sum_i d_i)$-time 
algorithm that given a sequence $d$ of integers proves that the given sequence is not graphic or 
computes an $m$-edge graph realizing it, where 
$m ~=~ \frac{1}{2} \sum_{i=1}^n d_i$.

A more challenging variant of the problem arises when the given sequence has many possible realizing graphs, but some realizations are more desirable
than others in various ways. In such a case, it is of interest to look for a realizing graph that also optimizes some 
quality measure. Hereafter, we refer to such problems as \emph{optimized realization} problems.

For example, let us consider the classical \cliquelong (\clique) problem. For a graph $G = (V,E)$, a clique is a vertex set $Q \subseteq V$ such that $(v,w) \in E$ for every $v, w \in Q$. The size of the maximum clique in $G$ is denoted 
$\clique(G) = \max \{|Q| \mid Q \text{ is a clique in } G\}$.
It may be desirable to find, for a given graphic sequence $d$, a realizing graph $G$ with the maximum possible clique. 
The resulting optimized realization problem is formally defined as follows. Letting $\clique(d) = \max \{\clique(G) \mid G \text{ is a realization of } d\}$, 
the \cliqueDRlong (\cliqueDR) problem requires to find a realizing graph $G$ attaining $\clique(d)$.

Similarly, one may define optimized realization problems corresponding to 
other graph optimization problems, 
including \islong, \vclong, 
\matchlong, 
and \domsetlong.
Their optimal values on a given graph $G$ are defined as $\is(G)$, $\vc(G)$, 
$\match(G)$, 
and $\domset(G)$, respectively. 
The corresponding optimized realization problems on degree sequences are named 
\isDRlong (\isDR), \vcDRlong (\vcDR), 
\matchDRlong (\matchDR), 
and \domsetDRlong (\domsetDR),
and their optimal values on a given sequence $d$ are denoted by 
$\is(d)$, $\vc(d)$, 
$\match(d)$, 
and $\domset(d)$.

As is well known, the graph versions of 
most of these optimization problems are NP-hard. 
In contrast, the optimized realization versions of the Maximum Clique, Minimum Vertex Cover 
and Maximum Independent Set problems have been shown to be polynomial-time solvable.
Rao~\cite{Rao79} gave a characterization of sequences which can be realized by a graph 
containing $K_\ell$ (i.e., a clique of size $\ell$). This result was based on 
a phenomenon collectively known as the \emph{prefix lemma}, which for MC says that
if $d$ has a realization containing $K_\ell$, then it has a realization such that 
the $\ell$ vertices of maximum degree induce a clique $K_\ell$.
This result implies a polynomial-time algorithm for \cliqueDR.
Recalling that \is is equivalent to \clique in the complement graph, and since the set of all the vertices that are not in an independent set are a vertex cover, it follows that polynomial-time algorithms also exist for \isDR and \vcDR.
Note that a prefix lemma also applies to \vc.
This prefix property is satisfied also by \domset~\cite{GHR18} and \match~\cite{GJL99}. 
However, so far it was not known how to exploit this property in order to derive a polynomial time 
algorithm for MM-DR or for MDS-DR. These two problems were handled in some special cases
in~\cite{BockRautenbach19} and in~\cite{GHR18} respectively, but the general problems were left open.

\paragraph{Our results.} 
This paper provides polynomial time realization algorithms for both \domsetDR and \matchDR. 
The algorithm for \domsetDR makes use of an existing prefix lemma~{\cite{GHR18}}
while the algorithm for \matchDR makes use of a stronger version of an existing prefix lemma~{\cite{GJL99}}, 
established in the current paper.

In addition, we develop Erd\H{o}s-Gallai like characterizations for \domsetDR. 
These characterizations can be used to efficiently compute the size of the minimum dominating set of a sequence without providing a realization by searching for the size of the minimum dominating set.
Interestingly, our characterizations for the \domsetDR problem are based on
the construction and correctness of our realization algorithm, 
and in turn, the characterizations help us reducing the complexity of the realization 
algorithm, because the realization algorithm does not need to search for the size of 
the minimum dominating set.
A similar complexity reduction is possible also for the \matchDR problem  
for which characterizations are already known~\cite{EKMT24}.

As mentioned above, there are several known algorithms for finding a graph $G$ realizing $d$, including the well-known Havel--Hakimi algorithm~\cite{Havel55,Hakimi62}.
For our purposes, though, we need to use a somewhat lesser-known but highly versatile algorithm due to Fulkerson, Hoffman and McAndrew~{\cite{FHM1965}}, hereafter named the FHM algorithm, which is based on 
\begin{inparaenum}[(i)]
\item 
realizing $(d,d)$ by a \emph{bipartite} graph $\hat{G}$,
\item converting $\hat{G}$ to a half integral general (non-bipartite) realization $G^\omega$ for $d$, and
\item rounding $G^\omega$ to an integral general realization $G$.
\end{inparaenum}
Our approach is based on modifying step (i) so that the bipartite realization $\hat{G}$ has optimal MDS (or MM), and preserving this property during the transformations of steps (ii) and (iii). The challenging obstacle is that the rounding process of step (iii) is oblivious to the issue of MDS (or MM) size, so a bipartite $\hat{G}$ with small MDS might be transformed into a general $G$ with large MDS. Hence, it is necessary to
modify the FHM algorithm in non-trivial ways in order to solve the MDS-DR and MM-DR optimized realization problems\footnote{The same approach can 
provide an alternative algorithm for \cliqueDR and for \vcDR 
(or \isDR).
We omit the details.}.


\paragraph*{Related work.}
One related direction involves network realization with a given subgraph. 
Kundu~\cite{Kundu73} showed that the sequences $d$ and
the component-wise difference $d - d'$, such that $d'_i \in \set{k,k+1}$ and $d' \le d$ (component-wise), 
are graphic only if there exists a graph $G$ realizing $d$ that has a subgraph $G'$ realizing $d'$. 
This result can be used to decide whether a given sequences has a perfect matching by assigning $d'_i = 1$ for every $i$.
Kleitman and Wang~\cite{KleitmanW73} gave an algorithm for computing a realization of $d$ that contains a subgraph which realizes $d'$. 
Their algorithm can be used to compute a realization of $d$ that contains a perfect matching, if 
one exists.
Extensions of the above result were presented in~\cite{KleitmanW73,Kundu74}.

Rao and Rao{~\cite{RaoRao72}} and Kundu{~\cite{Kundu73}} gave a characterization of 
sequences that can be realized by a Hamiltonian graph.
Chungphaisan~\cite{Chungphaisan78} gave an algorithm, that given a sequence $d$,
constructs a realization with a Hamiltonian cycle (or path), if one
exists.
Rao~\cite{Rao79} characterized sequences that can be realized by a graph containing $K_\ell$ (a clique of size $\ell$). This result was based on the Prefix Lemma for MC.
An alternative proof was given in~\cite{KezdyLehel96}, 
and a constructive proof was presented in~{\cite{Yin12}}.
It follows that \cliqueDR can be solved in polynomial time by 
checking if $d$ has a realization containing $K_\ell$, for 
$\ell \in [2,n]$, where $[i,j] = \{i,\ldots,j\}$, for $i \leq j$.
In contrast, it is NP-hard to approximate \clique within a ratio of $O(n^{1-\varepsilon})$, 
for any $\varepsilon > 0$~\cite{Zuckerman07}.
Gould, Jackson and Lehel~\cite{GJL99} extended Rao's result by showing that if $d$ 
has a realization containing $H$ as a subgraph (but not necessarily an induced sub-graph),
then there exists a realization of $d$ containing $H$ such that the vertices of $H$ 
have the $|V(H)|$ largest degrees.
Yin~\cite{Yin11} used this result to further extend the result of Rao by giving a characterization of graphic sequences that 
contains a split graph $S_{r,s}$ composed of a clique is of size $r$ and an independent set of size $s$. 
(Observe that $S_{r,1}$ is a clique of size $r+1$.)

Gentner, Henning and Rautenbach~\cite{GHR18} 
proved the prefix lemma for \domsetDR and gave a realization algorithm for 
\domsetDR on sequences with $d_1 = O(1)$, leaving the general case open.
They also provided characterizations for \isDR and \domsetDR in forests.
Gentner, Henning and Rautenbach~{\cite{GHR16}} gave characterizations to 
realizations that minimize the maximum independent set and that maximize the minimum 
dominating set in forests.
Note that
\domset is not approximable within $\alpha \log n$, for some $\alpha>0$,
unless $\text{P}=\text{NP}$~{\cite{RazSafra97}}, and within $(1-\varepsilon) \log n$, for any
$\varepsilon > 0$, unless $\text{NP} \subseteq \text{DTIME}(n^{\log \log
  n})$~{\cite{Feige98}}.

Bock and Rautenbach~\cite{BockRautenbach19} studied \matchDR in trees and bipartite graphs, where the partition is given.
The result on bipartite graphs is based on a stronger version of the prefix lemma 
that focuses on a specific matching (see~{\Cref{lem: MM prefix}}).
Recently, Erd\"{o}s et al.~\cite{EKMT24} studied \matchDR. 
Given a sequence $d$ and an integer $\nu$, they presented an Erd\H{o}s-Gallai type characterization, based on a system of $O(n)$ inequalities,
which is satisfied if and only if $d$ has a realization with a matching of size $\nu$.
Applying these characterizations to a given degree sequence, the maximum size of a 
matching can be computed efficiently, although the specific realization is not provided.

Fulkerson, Hoffman and McAndrew~\cite{FHM1965} obtained conditions for the existence of 
an $f$-factor that are applicable only to the family of multi-graphs that satisfy the 
so called \emph{odd cycle condition}.
Kundu~\cite{Kundu74a} used this result to provide simplified conditions for 
the factorization of graphs that satisfy the odd cycle condition.
Rao~\cite{Rao81towards} and Yin~{\cite{Yin11}} also used the result of~\cite{FHM1965}.
Anstee used the technique of~\cite{FHM1965} to provide an algorithmic proof of the $f$-factor theorem~\cite{Anstee85}
and a simplified characterization for the existence of a $(g,f)$-factor for special cases~\cite{Anstee90}.

\section{Realization with Minimum Dominating Set}
\label{s:MDS-DR alg}

In this section, we describe an algorithm for constructing a realization 
of a given graphic sequence $d$, which in addition has a dominating set $D$ of the minimum size $\gamma$ among all the possible realizations.
For this, we employ the Prefix Lemma \ref{lem: MDS prefix}  \cite{GHR18} and a suitable modification of the FHM realization algorithm
\cite{FHM1965}.

Our algorithm proceeds in several steps, presented in the coming subsections.
First, the MDS-DR problem over general graph is reduced to the same problem over bipartite graphs. 
It is done by a modification of the FHM realization algorithm 
that ensures preservation of a dominating set.
Then, the MDS-DR problem over bipartite graphs is reduced to a maximum flow problem. 
To be more specific, given a candidate size $\gamma$ for the minimum dominating set and a degree sequence pair $(d,d)$, 
we construct a bipartite flow graph $G_{d,\gamma}$, 
such that if the maximum flow attained in $G_{d,\gamma}$ equals $\sum_{i=1}^n d_i$, 
then it corresponds to a realization $\hat G$ of $(d,d)$ with a dominating set of size $2\gamma$. 
The Prefix \Cref{lem: MDS prefix} narrows down the search of the dominating set to a polynomial number of candidates, 
which allows solving MDS-DR in polynomial time.

\subsection{Prefix Lemma}

Define a \emph{$\gamma$-prefix-dominated realization} of the sequence $d$ to be a realization, 
where the vertices with the $\gamma$ highest degrees (i.e., $d_1, d_2 \ldots, d_\gamma$) form a dominating set.

\begin{lemma}
{\bf (Prefix Lemma for \domset) \cite{GHR18}} 
\label{lem: MDS prefix}
If a sequence $d$ has a realization with a minimum dominating set of size $\gamma$, then $d$ has $\gamma$-prefix-dominated realization.
\end{lemma}

Although the prefix lemma states that if $\gamma = \domset(d)$ then there is a realization with dominating set on $\gamma$ vertices of highest degrees,
it is not immediately clear how to select edges so as to obtain this realization, hence additional ideas are needed.
\subsection{Reduction to a Bipartite Sequence Pair} 
A bipartite graph $\hat{G} = (V, W, \hat{E})$, where $V = \{v_1, v_2, \ldots, v_n\}$ and $W = \{w_1, w_2, \ldots, w_n\}$, is a \emph{$\gamma$-prefix-dominated} realization for the sequence pair $(d, d)$ if it satisfies the following properties.
\begin{compactenum}[(D1)]
\item $\hat{G}$ realizes the sequence pair $(d,d)$.
\item $(v_i, w_i) \notin \hat{E}$ for every $i \in [1,n]$.
\item $\hat{D} = \hat{D}_V \cup \hat{D}_W$ is a dominating set in $\hat{G}$, 
      where $\hat{D}_V = \{v_1, \ldots, v_\gamma\}$ and $\hat{D}_W= \{w_1, \ldots, w_\gamma\}$ 
      are prefixes of $V$ and $W$, respectively.
\end{compactenum}

We describe a polynomial time algorithm based on the FHM algorithm~\cite{FHM1965}, 
that given a $\gamma$-prefix dominated realization $\hat G$ for the sequence pair $(d,d)$ produces a $\gamma$-prefix dominated realization $G$ for $d$. 

\inline Step 1: 
Compute a half-integral solution.
\begin{enumerate}[a.]
\item
For all $i, j \in [1,n]$, let ~~~~
$\displaystyle
y_{ij} = \begin{cases}
    1, ~ \{v_i, w_j\} \in \hat{E}, \\
    0, ~ \text{otherwise,}
\end{cases}$
~~~~
and ~~
$\omega(i,j) = \frac{1}{2}(y_{ij} + y_{ji})$.
\item 
Define a weighted graph $G^\omega = (V^\omega, E^\omega,\omega)$ with vertex set $V^\omega = [1,n]$ and 
an edge $e=(i, j)$ of weight $\omega(e)=$ $\omega(i,j)$ for every $i,j \in V^\omega$. Clearly, $w$ is half-intergral.
\item
Define the \emph{weighted degree} of a vertex $i \in V^\omega$ to be $d^\omega(i) = \sum_{j\in V^\omega} \omega(i, j)$. Note that $G^\omega$ realizes $d$ in the \emph{weighted} sense, namely,
\\
\hbox{\hskip 20pt}
$
d^\omega(i) = \sum_{j\in V^\omega} \omega(i, j) = \frac{1}{2} \left(\sum_{j=1}^n y_{ij} + \sum_{j=1}^n y_{ji}\right) = d_i, ~ \mbox{for any}~ i \in V^\omega.~~~~~~~\mbox{}$
\item
Partition the vertex set $V^\omega$ of $G^\omega$ into $D = [1,\gamma]$ and $S = [\gamma+1,n]$. Note that by construction, 
$D$ is a dominating set for $G^\omega$, and moreover, the total weight of the edges connecting any nondominating vertex in $S$ with its dominating neighbours in $D$ is at least one. Indeed, for any $s \in S$, by property (D3),
\begin{align} 
\label{eq: dominating weight}
\textstyle
\sum_{x \in D} \omega(s, x) 
~=~ \sum_{j \in [1, \gamma]} \omega(s,j) 
~=~ \frac{1}{2} \sum_{j \in [1, \gamma]} y_{sj} +  \frac{1}{2} \sum_{j \in [1, \gamma]} y_{js} 
~\ge~ \frac{1}{2} + \frac{1}{2} ~=~ 1.
\end{align}
\end{enumerate}

\inline Step 2: 
Preparing for discarding non-integral weights while keeping the degrees.
\\
Construct a graph $\Ghalf = (\Vhalf,\Ehalf)$ 
by removing from $G^\omega$ the edges of integral weight and keeping only those of weight $1/2$.
Formally, $\Vhalf = V^\omega$ and $\Ehalf = \{e \in E^\omega \mid \omega(e) = 1/2 \}$.

\begin{observation} 
\label{obs: G1/2 is even graph}
The graph $G^{1/2}$ is even
(namely, all its vertex degrees are even).
\end{observation}
\begin{proof}
Each vertex $i \in V^\omega$ has an integral weighted degree $d^\omega(i)$, so the number of edges of weight $1/2$ incident to $i$ must be even.
\end{proof}

The first two steps are similar to the FHM algorithm \cite{FHM1965}, and the main changes w.r.t. that algorithm occur in the subsequent steps, and pertain to the process of modifying $G^\omega$ and getting rid of non-integral weights while keeping the degrees unchanged without violating Inequality \eqref{eq: dominating weight}. The modifications happen concurrently on the graphs $G^\omega$ and $\Ghalf$
(i.e., whenever changing the weight of some edge $e$ in $G^\omega$ from $0$ or $1$ to $1/2$ or vice versa, $\Ghalf$ is modified accordingly, adding or removing the edge $e$).

Specifying the modifications require the following definition. A 4-vertex path $P[a, s, b, c]$ in $\Ghalf$ is a \emph{2-dom path} if $s \in S$ and $a, b \in D$, i.e., the nondominating $s$ has two neighboring dominators. 

\inline Step 3: 
Eliminate 2-dom paths.
\\
The next step in the algorithm transforms the weights in $G^\omega$ until it is free of 2-dom paths. This is done as follows.
While there is a 2-dom path 
in $\Ghalf$, apply one of the following three \emph{modification rules} to the edge weights in $G^\omega$, according to a weight of the edge $(a, c)$ in $G^\omega$. The different possible situations are visualized in Figure \ref{fig: modifications}. All figures in this section maintain the convention that black nodes are dominating, white nodes are nondominating and gray nodes can be either.
Furthermore, solid lines represent edges of positive weight and dashed lines represent edges of weight 0.

\begin{enumerate}
\item[(MR1)] 
If $\omega(a,c) = 0$, 
then set $\omega(a, s) \gets 0$, $\omega(s,b) \gets 1$, $\omega(b,c) \gets 0$ and $\omega(a,c) \gets 1/2$.
    
\item[(MR2)] 
If $\omega(a,c) = 1/2$, then set $\omega(a,s) \gets 1$, $\omega(s,b) \gets 0$, $\omega(b,c) \gets 1$ and $\omega(a,c) \gets 0$.

\item[(MR3)] 
If $\omega(a,c) = 1$, then set $\omega(a,s) \gets 1$, $\omega(s,b) \gets 0$, $\omega(b,c) \gets 1$ and $\omega(a,c) \gets 1/2$.

\end{enumerate}

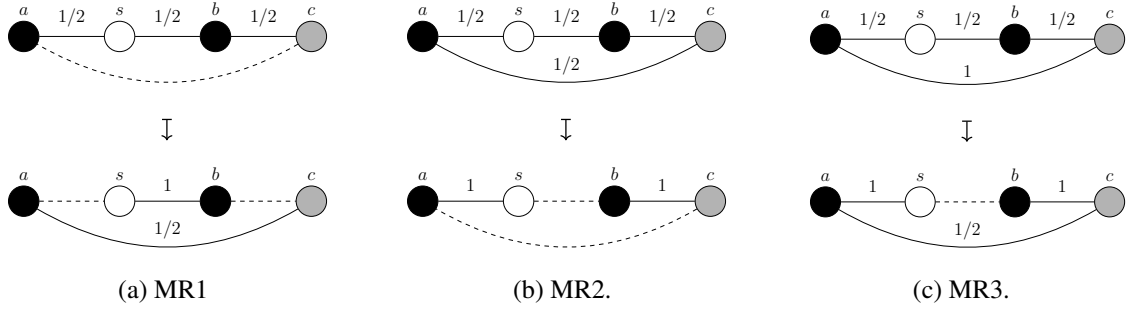
\begin{figure}[t]
    \centering
\begin{tabular}{c@{\hspace{10pt}}c@{\hspace{10pt}}c}
    \begin{subfigure}{.3\textwidth}
    \centering
    \resizebox{0.9\textwidth}{!}{
    \tikzset{
    myptr/.style={-{Stealth[scale=1.5]}},
    }
    
    \definecolor{gr}{rgb}{0.7, 0.7, 0.7}
    
    \begin{tikzpicture}[
        graynode/.style={circle, draw=black, fill=gr, minimum size=20pt, inner sep=0pt},
        blacknode/.style={circle, draw=black, fill=black, minimum size=20pt, inner sep=0pt},
        whitenode/.style={circle, draw=black, minimum size=20pt, inner sep=0pt},
        font=\large
    ]

    \node[whitenode, label={$s$}] (s) {};
    \node[blacknode, label={$a$}, left=1.5cm of s] (a) {};
    \node[blacknode, label={$b$}, right=1.5cm of s] (b) {};
    \node[graynode, label={$c$}, right=1.5cm of b] (c) {};
    \node[right=0.6cm of s] (sup) {};
    
    \draw (a) -- (s) node[pos=0.5, above=0.1cm] {$1/2$};
    \draw (s) -- (b) node[pos=0.5, above=0.1cm] {$1/2$};
    \draw (b) -- (c) node[pos=0.5, above=0.1cm] {$1/2$};
    \draw[dashed] (c) to[out=-150,in=-30] (a);

    \node[below=1.5cm of sup] (m1) {};
    \node[below=0.5cm of m1] (m2) {};
    \draw[|->, line width=0.3mm]  (m1)   -- (m2);

    \node[below=1cm of m2] (sup1) {};
    \node[whitenode, label={$s$}, left=0.6cm of sup1] (s1) {};
    \node[blacknode, label={$a$}, left=1.5cm of s1] (a1) {};
    \node[blacknode, label={$b$}, right=1.5cm of s1] (b1) {};
    \node[graynode, label={$c$}, right=1.5cm of b1] (c1) {};

    \draw[dashed] (a1) -- (s1);
    \draw (s1) -- (b1) node[pos=0.5, above=0.1cm] {$1$};
    \draw[dashed] (b1) -- (c1);
    \draw (c1) to[out=-150,in=-30] node[midway, above=0.05cm] {$1/2$}  (a1);
    
    \end{tikzpicture}
    }
    \caption{MR1}
    \label{fig: modification one}
    \end{subfigure}
&
    \begin{subfigure}{.3\textwidth}
    \centering
    \resizebox{0.9\textwidth}{!}{
    \tikzset{
    myptr/.style={-{Stealth[scale=1.5]}},
    }
    
    \definecolor{gr}{rgb}{0.7, 0.7, 0.7}
    
    \begin{tikzpicture}[
        graynode/.style={circle, draw=black, fill=gr, minimum size=20pt, inner sep=0pt},
        blacknode/.style={circle, draw=black, fill=black, minimum size=20pt, inner sep=0pt},
        whitenode/.style={circle, draw=black, minimum size=20pt, inner sep=0pt},
        font=\large
    ]

    \node[whitenode, label={$s$}] (s) {};
    \node[blacknode, label={$a$}, left=1.5cm of s] (a) {};
    \node[blacknode, label={$b$}, right=1.5cm of s] (b) {};
    \node[graynode, label={$c$}, right=1.5cm of b] (c) {};
    \node[right=0.6cm of s] (sup) {};
    
    \draw (a) -- (s) node[pos=0.5, above=0.1cm] {$1/2$};
    \draw (s) -- (b) node[pos=0.5, above=0.1cm] {$1/2$};
    \draw (b) -- (c) node[pos=0.5, above=0.1cm] {$1/2$};
    \draw (c) to[out=-150,in=-30] node[pos=0.5, above=0.05cm] {$1/2$} (a);

    \node[below=1.5cm of sup] (m1) {};
    \node[below=0.5cm of m1] (m2) {};
    \draw[|->, line width=0.3mm] (m1)   -- (m2);
    
    \node[below=1cm of m2] (sup1) {};
    \node[whitenode, label={$s$}, left=0.6cm of sup1] (s1) {};
    \node[blacknode, label={$a$}, left=1.5cm of s1] (a1) {};
    \node[blacknode, label={$b$}, right=1.5cm of s1] (b1) {};
    \node[graynode, label={$c$}, right=1.5cm of b1] (c1) {};

    \draw (a1) -- (s1) node[pos=0.5, above=0.1cm] {$1$};
    \draw[dashed] (s1) -- (b1);
    \draw (b1) -- (c1) node[pos=0.5, above=0.1cm] {$1$};
    \draw[dashed] (c1) to[out=-150,in=-30] (a1);
    
    \end{tikzpicture}
    }
    \caption{MR2.}
    \label{fig: modification two}
    \end{subfigure}
&
    \begin{subfigure}{.3\textwidth}
    \centering
    \resizebox{0.9\textwidth}{!}{
    \tikzset{
    myptr/.style={-{Stealth[scale=1.5]}},
    }
    
    \definecolor{b}{rgb}{0.0, 0, 1}
    \definecolor{gr}{rgb}{0.7, 0.7, 0.7}
    
    \begin{tikzpicture}[
        graynode/.style={circle, draw=black, fill=gr, minimum size=20pt, inner sep=0pt},
        blacknode/.style={circle, draw=black, fill=black, minimum size=20pt, inner sep=0pt},
        whitenode/.style={circle, draw=black, minimum size=20pt, inner sep=0pt},
        font=\large
    ]

    \node[whitenode, label={$s$}] (s) {};
    \node[blacknode, label={$a$}, left=1.5cm of s] (a) {};
    \node[blacknode, label={$b$}, right=1.5cm of s] (b) {};
    \node[graynode, label={$c$}, right=1.5cm of b] (c) {};
    \node[right=0.6cm of s] (sup) {};
    
    \draw (a) -- (s) node[pos=0.5, above=0.1cm] {$1/2$};
    \draw (s) -- (b) node[pos=0.5, above=0.1cm] {$1/2$};
    \draw (b) -- (c) node[pos=0.5, above=0.1cm] {$1/2$};
    \draw (c) to[out=-150,in=-30] node[pos=0.5, above=0.05cm] {$1$} (a);

    \node[below=1.5cm of sup] (m1) {};
    \node[below=0.5cm of m1] (m2) {};
    \draw[|->, line width=0.3mm] (m1)   -- (m2);
    
    \node[below=1cm of m2] (sup1) {};
    \node[whitenode, label={$s$}, left=0.6cm of sup1] (s1) {};
    \node[blacknode, label={$a$}, left=1.5cm of s1] (a1) {};
    \node[blacknode, label={$b$}, right=1.5cm of s1] (b1) {};
    \node[graynode, label={$c$}, right=1.5cm of b1] (c1) {};

    \draw (a1) -- (s1) node[pos=0.5, above=0.1cm] {$1$};
    \draw[dashed] (s1) -- (b1);
    \draw (b1) -- (c1) node[pos=0.5, above=0.1cm] {$1$};
    \draw (c1) to[out=-150,in=-30] node[pos=0.5, above=0.05cm] {$1/2$} (a1);
    
    \end{tikzpicture}
    }
    \caption{MR3.}
    \label{fig: modification three}
    \end{subfigure}
\end{tabular}
\vspace{-5pt}
\caption{Illustration to the modification rules.}
\label{fig: modifications}
\end{figure}

Note that the modifications preserve the weighted degree of every vertex in $G^\omega$ and the total weight of the edges between any nondominating vertex and its dominating neighbors. 
Moreover, each time a modification is applied, the number of edges in $\Ghalf$ decreases. Therefore, the 2-dom paths are eliminated from $\Ghalf$ (with the corresponding paths eliminated from $G^\omega$) within a polynomial number of steps. 

\inline Step 4: Separate special cycles.
\\
At the start of this step, $G^{1/2}$ is free of 2-dom paths.
Let $S'$ be the set of all the nondominating vertices in $G^\omega$ that are connected to a dominator vertex with at least one edge $\hat{e}$ with $\omega(\hat{e})=1$. Choose such an edge $\hat{e}_s$ arbitrarily for each vertex $s \in S'$ and denote the collection of these edges by $E' = \{\hat{e}_s \mid s \in S'\}$.

Next consider the set $S^\Delta = S \setminus S'$ of the remaining nondominating vertices. Since every $s \in S^\Delta$ is not connected to any vertex in $D$ with an edge of weight one in $G^\omega$, it has at least two dominating neighbours in $\Ghalf$ by 
Inequality~\eqref{eq: dominating weight}. Choose arbitrarily two such neighbours $a_s, b_s \in D$ for every $s \in S^\Delta$.

\begin{observation}
\label{obs: poor connection}
For every $s \in S^\Delta$, $a_s$ and $b_s$ are not connected to any vertices in $\Ghalf$ besides possibly each other and $s$.
\end{observation}

\begin{proof}
Having another vertex $c$ neighboring, say, $b_s$, in $G^{1/2}$ would imply the existence of a 2-dom path $P[a_s, s, b_s, c]$ in $G^{1/2}$, leading to a contradiction.
\end{proof}

In all modifications performed in subsequent steps of the algorithm, edges are only removed from (but never added to) $\Ghalf$, 
so \Cref{obs: poor connection} continues to hold until the end of the algorithm's execution.

Since $a_s$ and $b_s$ have integral weighted degrees in $G^\omega$, {\Cref{obs: poor connection}} implies that the edge $(a_s, b_s)$ must exist in $\Ghalf$. It follows that there is a cycle $C[a_s, s, b_s]$ in $\Ghalf$ for each $s \in S^\Delta$ (see {\Cref{fig:s-cycle}}). Let $\calC^\Delta = \{C[a_s, s, b_s] \mid s \in S^\Delta\}$ be a set of all such cycles.

    \begin{figure}[H]
    \centering
    \resizebox{0.17\textwidth}{!}{
    \tikzset{
    myptr/.style={-{Stealth[scale=1.5]}},
    }
    
    \definecolor{b}{rgb}{0.0, 0, 1}
    \definecolor{gr}{rgb}{0.7, 0.7, 0.7}
    
    \begin{tikzpicture}[
        graynode/.style={circle, draw=black, fill=gr, minimum size=20pt, inner sep=0pt},
        blacknode/.style={circle, draw=black, fill=black, minimum size=20pt, inner sep=0pt},
        whitenode/.style={circle, draw=black, minimum size=20pt, inner sep=0pt},
        font=\Large
    ]

    \node[whitenode, label={$s$}] (s) {};
    \node[below = 2cm of s] (sup) {};
    \node[blacknode, label={$a$}, left=1cm of sup] (a) {};
    \node[blacknode, label={$b$}, right=1cm of sup] (b) {};
    
    \draw (a) -- (s) node[pos=0.5, above left=0.1cm] {$\frac{1}{2}$};
    \draw (s) -- (b) node[pos=0.5, above right=0.1cm] {$\frac{1}{2}$};
    \draw (a) -- (b) node[pos=0.5, above =0.1cm] {$\frac{1}{2}$};
    \end{tikzpicture}
    }
    \vspace{-5pt}
    \caption{A cycle in $\calC^\Delta$.}
    \label{fig:s-cycle}
    \end{figure}
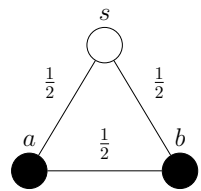

\begin{observation} 
\label{obs: disjoint cycles}
Any two different cycles $C, C'$ $\in \calC^\Delta$ have disjoint vertices.
\end{observation}

\begin{proof}
Consider cycles $C=C[s, a_s, b_s]$ and $C'=C[s', a'_s, b'_s]$. By construction, $s \neq s'$, so 
$C$ and $C'$ 
can only intersect in $a_s$ or $b_s$, but this cannot happen
by \Cref{obs: poor connection}. 
\end{proof}

\inline Step 5:
Partition into cycles.
\\
Denote the set of all the edges of the cycles in $\calC^\Delta$ by 
\begin{align*}
E(\calC^\Delta) = \{(a_s, b_s), (a_s, s), (s, b_s) \mid C[a_s, s, b_s] \in \calC^\Delta\}
\end{align*}
Consider a subgraph $H$ of $\Ghalf$ with the same vertices, $V(H) = \Vhalf$, and edges $E(H) = \Ehalf \setminus E(\calC^\Delta)$. Since $\Ghalf$ is an even graph by Observation \ref{obs: G1/2 is even graph} and the cycles in $\calC^\Delta$ are disjoint by Observation \ref{obs: disjoint cycles}, it follows that $H$ is an even graph. So there is a polynomial time algorithm for partitioning edge set of $H$ into disjoint cycles, such that each cycle contains an entire connected component. Denote the set of these cycles by $\calC'$ and let $\calC = \calC' \cup \calC^\Delta$. We call a cycle \emph{even} (resp., \emph{odd}) if it has an even (resp., odd) number of edges.

\begin{observation} 
\label{obs: even number of odd cycles-2}
The number of odd cycles in $\calC$ is even.
\end{observation}
\begin{proof}
Observe that 
$\sum_{e \in E^\omega} \omega(e) ~=~ \frac{1}{2} \sum_{i=1}^n d_i = m,$
where $m$ is the number of edges, which is an integer.
Therefore, the number of edges with weight $1/2$ must be even. 
Since the cycles in $\calC$ cover all of $\Ehalf$ and are disjoint, the observation follows.
\end{proof}

\begin{observation} \label{obs: intersection of C' and C Delta}
If cycles $C \in \calC^\Delta$ and $C' \in \calC'$ have a non-empty intersection, then $C \cap C' = \{s\}$ for $s \in S^\Delta$ corresponding to the cycle $C$.
\end{observation}
\begin{proof}
Let $C = C[a_s, s, b_s]$ with $a_s, b_s \in D$ and $s \in S^\Delta$.
By \Cref{obs: poor connection}, $a_s$ and $b_s$ have degree zero in $H$, so they do not belong to any cycle in $\calC'$. 
\end{proof}

\inline Step 6:
Eliminate even cycles.
\\
For every even cycle $C \in \calC$ do the following.
\begin{compactenum}
\item Traverse $C$ starting from an arbitrary vertex $x \in C$ and continuing along the cycle until returning to $x$. Denote  the resulting sequence of edges by $E(C)=(e_1, e_2, \ldots, e_\ell)$.
\item Increase (respectively, decrease) the weights of the edges on even (resp., odd) positions in the sequence $E(C)$ by $1/2$.
That is, for every $i \in [1,\ell]$ set
\[
\omega(e_i) \gets 
\begin{cases}
1, ~ \text{$i$ is even}, \\
0, ~ \text{$i$ is odd}.
\end{cases} 
\]
\end{compactenum}
Note that this procedure does not change the weighted degrees in $G^\omega$ or the weights of the edges in $E'$ and does not affect other cycles in $\calC$. Hereafter, we refer to such modifications as \emph{neutral}. 

\inline Step 7:
Eliminate odd cycles.
\\
Finally, arrange the odd cycles in $\calC$ (whose number is even by \Cref{obs: even number of odd cycles-2}) in pairs. For every pair $(C, C')$, proceed according to Case~1 below if the cycles intersect and according to Case~2 otherwise.

\inline Case 1: There is a vertex $x \in C \cap C'$:
\begin{compactenum}
\item 
As before, traverse both cycles starting at $x$. Denote the resulting sequences of edges in $C$ and $C'$ by $E(C) = (e_1, e_2, \ldots, e_\ell)$ and $E(C')=(e'_1, e'_2, \ldots, e'_k)$, respectively.
\item 
Cycles in $\calC'$ are disjoint by definition and cycles in $\calC^\Delta$ are disjoint by \Cref{obs: disjoint cycles}. 
Thus one of the cycles $(C,C')$ belongs to $\calC'$ and the other to $\calC^\Delta$. Redefine them so $C \in \calC^\Delta$ and $C' \in \calC'$. 
\item 
For every $i \in [1,\ell]$ and $j \in [1,k]$, modify the edge weights in the cycles as follows:
\begin{align*}
\omega(e_i) & \gets \begin{cases}
            0, ~ \text{$i$ is even}, \\
            1, ~ \text{$i$ is odd}.
        \end{cases} 
&
\omega(e'_j) & \gets \begin{cases}
            1, ~ \text{$j$ is even}, \\
            0, ~ \text{$j$ is odd}.
        \end{cases} 
    \end{align*}
\end{compactenum}

Note that this modification is neutral. 

By \Cref{obs: intersection of C' and C Delta}, $x \in S^\Delta$ and $C = C[x, a_x, b_x]$ for some $a_x, b_x \in D$. Note that $x$ is connected to $a_x$ and $b_x$ in $G^\omega$ with edges of weight one after the modification. Neither $a_x$ nor $b_x$ belongs to any other cycles in $\calC$ by \Cref{obs: poor connection}, so the weights $\omega(x, a_x)$ and $\omega(x, b_x)$ are not modified further by the algorithm. Let $E^\Delta_1$ contain an edge $(x, a_x)$ for every cycle $C \in \calC^\Delta$ that was processed in this case.

\inline Case 2: $C \cap C' = \emptyset$:

First, we have the following lemma.

\begin{lemma} 
\label{lem: right x and y}
For any 
two
disjoint odd cycles $C, C' \in \calC$, there exist $x \in C$ and $y \in C'$, such that $(x,y) \notin E'$, $\omega(x, y) \neq 1/2$, and if $C$ or $C'$ belongs to $\calC^\Delta$, then the corresponding vertex does not belong to $S^\Delta$.
\end{lemma}

\begin{proof}
Recall that each cycle in $\calC^\Delta$ contains exactly one vertex from $S^\Delta$.
Choose vertices $x \in C$ and $y' \in C'$, so if $C \in \calC^\Delta$ (respectively, $C' \in \calC^\Delta$), then $x \notin S^\Delta$ (respectively, $y' \notin S^\Delta$). If $(x, y') \notin E'$, then the chosen $x$ and $y = y'$ satisfy the Lemma (the condition $w(x,y) \neq 1/2$ is proved later). Otherwise, one of the vertices $x$, $y'$ must be dominating and the other nondominating. Without loss of generality assume $y' \in D$ and $x \in S$. Replace $y'$ with $y \in C'$, such that $y \neq y'$ and if $C' \in \calC^\Delta$, then $y \notin S^\Delta$. Since each cycle has at least three vertices, this is always possible. Each edge in $E'$ corresponds to a unique nondominating vertex and $(x, y')$ corresponds to $x$, so $E'$ does not contain $(x, y)$.
    
Next we prove that $\omega(x, y) \neq 1/2$. 
First, 
assume that  $C, C' \in \calC'$. Then by construction, they are not connected by any edge in $\Ehalf \setminus E(\calC^\Delta)$. 
Assume, towards contradiction, that they are connected by an edge $e\in E(\calC^\Delta)$, so $e$
belongs to a cycle $C[s, a_s, b_s]$ with $a_s, b_s \in D$ and $s \in S^\Delta$. 
It follows that either $C$ or $C'$ contains either $a_s$ or $b_s$. 
However it is not possible by \Cref{obs: intersection of C' and C Delta}, contradiction. 
Now assume 
that $C = C[s, a_s, b_s] \in \calC^\Delta$ with $a_s, b_s \in D$ and $s \in S^\Delta$. 
Since $x\in C \setminus S^\Delta$, $x$ is either $a_s$ or $b_s$. However $a_s$ and $b_s$ do not have any edges outside $C$ in $\Ghalf$ 
by \Cref{obs: poor connection}, contradiction. The lemma follows.
\end{proof}

Next we describe how to modify a pair $(C, C')$.
\begin{compactenum}
\item 
Choose vertices $x \in C$ and $y \in C'$ according to Lemma \ref{lem: right x and y}.
\item 
Traverse both cycles starting in $x$ and $y$ accordingly and denote the resulting sequences of edges in $C$ and $C'$ by $E(C) = (e_1, e_2, \ldots, e_\ell)$ and $E(C')=(e'_1, e'_2, \ldots, e'_k)$ respectively.
\item 
It follows that 
$\omega(x, y) \in \set{0,1}$.
Let  $\xi = \omega(x, y)$ and perform the following.
Set $\omega(x,y) \gets 1 - \xi$ and for every $i \in [1,\ell]$ and $j \in [1,k]$ modify the edge weights in the cycles as follows 
\begin{align*}
\omega(e_i) & \gets 
    \begin{cases}
        1-\xi, ~ \text{$i$ is even}, \\
        \xi, ~ \text{$i$ is odd},
    \end{cases} 
&
\omega(e'_j) & \gets 
    \begin{cases}
        1-\xi, ~ \text{$j$ is even}, \\
        \xi, ~ \text{$j$ is odd}.
    \end{cases} 
\end{align*}
\end{compactenum}

Note that this modification is neutral. 

If $\calC^\Delta$ contains $C$ or $C'$, then the corresponding $s \in S^\Delta$ is connected to a dominator vertex with an edge $\hat{e}_s$, such that $\omega(\hat{e}_s) = 1$ after applying the modification. Indeed, $s$ was not chosen as a starting vertex $x$ or $y$, so it is connected to one of its two neighbors $a_s$ or $b_s$ in the corresponding cycle with an edge of weight one. But both its neighbors are dominating vertices, so $\hat{e}_s$ exists. Note that neither $a_s$ nor $b_s$ belongs to any other cycles in $\calC$ by Observation \ref{obs: poor connection}, so the weight $\omega(\hat{e}_s)$ is not modified further by the algorithm. Let $E^\Delta_2$ contain all edges $\hat{e}_s$ for $s \in S^\Delta$ that were processed in this case.

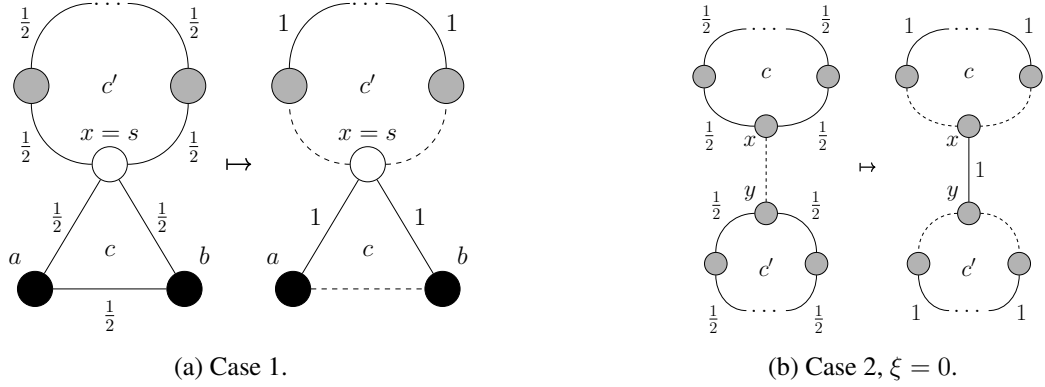
\begin{figure}[h!]
    \centering
\begin{tabular}{c@{\hspace{30pt}}c}
    \begin{subfigure}{.47\textwidth}
    \centering
    \resizebox{0.85\textwidth}{!}{
    \tikzset{
    myptr/.style={-{Stealth[scale=1.5]}},
    }
    
    \definecolor{b}{rgb}{0.0, 0, 1}
    \definecolor{gr}{rgb}{0.7, 0.7, 0.7}
    
    \begin{tikzpicture}[
        graynode/.style={circle, draw=black, fill=gr, minimum size=20pt, inner sep=0pt},
        blacknode/.style={circle, draw=black, fill=black, minimum size=20pt, inner sep=0pt},
        whitenode/.style={circle, draw=black, minimum size=20pt, inner sep=0pt},
        mynode/.style={minimum size=20pt, inner sep=0pt},
        font=\Large
    ]

    \node[whitenode, label={$x=s$}] (s) {};
    \node[below = 2cm of s] (sup) {};
    \node[blacknode, label={[xshift=-0.4cm]$a$}, left=1cm of sup] (a) {};
    \node[blacknode, label={[xshift=0.4cm]$b$}, right=1cm of sup] (b) {};

    \node[graynode, above left=1.5cm of s] (ca) {};
    \node[graynode, above right=1.5cm of s] (cb) {};
    \node[mynode, above=2.5cm of s] (dots) {$\dots$};

    \node[mynode, below=1cm of s] (c) {$c$};
    \node[mynode, above=0.9cm of s] (c') {$c'$};
    
    \draw (a) -- node[pos=0.6, left=0.15cm] {$\frac{1}{2}$} (s) ;
    \draw (b) -- node[pos=0.6, right=0.15cm] {$\frac{1}{2}$} (s);
    \draw (a) -- node[pos=0.5, below=0.1cm] {$\frac{1}{2}$} (b);

    \draw (s) to[out=180,in=-90] node[pos=0.5, left=0.2cm] {$\frac{1}{2}$} (ca);
    \draw (s) to[out=0,in=-90] node[pos=0.5, right=0.2cm] {$\frac{1}{2}$} (cb);
    \draw (ca) to[out=90,in=180] node[pos=0.5, left=0.2cm] {$\frac{1}{2}$} (dots);
    \draw (cb) to[out=90,in=0] node[pos=0.5, right=0.2cm] {$\frac{1}{2}$} (dots);

    \node[right=1.7cm of s] (m1) {};
    \node[right=0.5cm of m1] (m2) {};
    \draw[|->, line width=0.3mm] (m1) -- (m2);
    
    \node[whitenode, label={$x=s$}, right = 1.7cm of m2] (s1) {};
    \node[below = 2cm of s1] (sup1) {};
    \node[blacknode, label={[xshift=-0.4cm]$a$}, left=1cm of sup1] (a1) {};
    \node[blacknode, label={[xshift=0.4cm]$b$}, right=1cm of sup1] (b1) {};

    \node[graynode, above left=1.5cm of s1] (ca1) {};
    \node[graynode, above right=1.5cm of s1] (cb1) {};
    \node[mynode, above=2.5cm of s1] (dots1) {$\dots$};

    \node[mynode, below=1cm of s1] (c1) {$c$};
    \node[mynode, above=0.9cm of s1] (c'1) {$c'$};
    
    \draw (a1) -- node[pos=0.6, left=0.15cm] {$1$} (s1) ;
    \draw (b1) -- node[pos=0.6, right=0.15cm] {$1$} (s1);
    \draw[dashed] (a1) -- (b1);

    \draw[dashed] (s1) to[out=180,in=-90] (ca1);
    \draw[dashed] (s1) to[out=0,in=-90] (cb1);
    \draw (ca1) to[out=90,in=180] node[pos=0.5, left=0.2cm] {$1$} (dots1);
    \draw (cb1) to[out=90,in=0] node[pos=0.5, right=0.2cm] {$1$} (dots1);
    
    \end{tikzpicture}
    }
    \caption{Case 1.}
    \label{fig: first case}
    \end{subfigure}
&
    \begin{subfigure}{.42\textwidth}
    \centering
    \resizebox{0.7\textwidth}{!}{
    \tikzset{
    myptr/.style={-{Stealth[scale=1.5]}},
    }
    
    \definecolor{b}{rgb}{0.0, 0, 1}
    \definecolor{gr}{rgb}{0.7, 0.7, 0.7}
    
    \begin{tikzpicture}[
        graynode/.style={circle, draw=black, fill=gr, minimum size=20pt, inner sep=0pt},
        blacknode/.style={circle, draw=black, fill=black, minimum size=20pt, inner sep=0pt},
        whitenode/.style={circle, draw=black, minimum size=20pt, inner sep=0pt},
        font=\huge
    ]
    
    \node[] (dots0) {$\dots$};
    \node[graynode, below left=1.5cm of dots0] (l) {};
    \node[graynode, below right=1.5cm of dots0] (r) {};
    \node[graynode, below =2.5cm of dots0, label={[xshift=-0.5cm, yshift=-1.1cm]$x$}] (x) {};
    \node[below=4cm of dots0] (sup) {};

    \node[graynode, label={[xshift=-0.5cm, yshift=-0.1cm]$y$}, below=2cm of x] (cm) {};
    \node[graynode, below right=1.5cm of cm] (ca) {};
    \node[graynode, below left=1.5cm of cm] (cb) {};
    \node[below=2.5cm of cm] (dots) {$\dots$};

    \node[below=0.9cm of dots0] (c) {$c$};
    \node[below=0.9cm of cm] (c') {$c'$};
    
    \draw (dots0) to[out=0,in=90] node[pos=0.5, above right=0.05cm] {$\frac{1}{2}$} (r);
    \draw (dots0) to[out=180, in=90] node[pos=0.5, above left=0.05cm] {$\frac{1}{2}$} (l);
    \draw (r) to[out=-90,in=0] node[pos=0.5, below right=0.05cm] {$\frac{1}{2}$} (x);
    \draw (l) to[out=-90, in=180] node[pos=0.5, below left=0.05cm] {$\frac{1}{2}$} (x);

    \draw[dashed] (x) -- (cm);
    
    \draw (cm) to[out=0,in=90] node[pos=0.5, above right=0.05cm] {$\frac{1}{2}$} (ca);
    \draw (cm) to[out=180, in=90] node[pos=0.5, above left=0.05cm] {$\frac{1}{2}$} (cb);
    \draw (ca) to[out=-90,in=0] node[pos=0.5, below right=0.05cm] {$\frac{1}{2}$} (dots);
    \draw (cb) to[out=-90, in=180] node[pos=0.5, below left=0.05cm] {$\frac{1}{2}$} (dots);

    \node[right=2.5cm of sup] (m1) {};
    \node[right=0.5cm of m1] (m2) {};
    \draw[|->, line width=0.4mm] (m1) -- (m2);
    
    \node[right=2.5cm of m2] (sup1) {};
    \node[above=4cm of sup1] (dots01) {$\dots$};
    \node[graynode, below left=1.5cm of dots01] (l1) {};
    \node[graynode, below right=1.5cm of dots01] (r1) {};
    \node[graynode, below =2.5cm of dots01, label={[xshift=-0.5cm, yshift=-1.1cm]$x$}] (x1) {};

    \node[graynode, label={[xshift=-0.5cm, yshift=-0.1cm]$y$}, below=2cm of x1] (cm1) {};
    \node[graynode, below right=1.5cm of cm1] (ca1) {};
    \node[graynode, below left=1.5cm of cm1] (cb1) {};
    \node[below=2.5cm of cm1] (dots1) {$\dots$};

    \node[below=0.9cm of dots01] (c1) {$c$};
    \node[below=0.9cm of cm1] (c'1) {$c'$};
    
    \draw (dots01) to[out=0,in=90] node[pos=0.5, above right=0.05cm] {$1$} (r1);
    \draw (dots01) to[out=180, in=90] node[pos=0.5, above left=0.05cm] {$1$} (l1);
    \draw[dashed] (r1) to[out=-90,in=0] node[pos=0.5, below right=0.05cm]{} (x1);
    \draw[dashed] (l1) to[out=-90, in=180] node[pos=0.5, below left=0.05cm] {} (x1);

    \draw (x1) -- node[pos=0.5, right=0.05cm] {$1$} (cm1);
    
    \draw[dashed] (cm1) to[out=0,in=90] (ca1);
    \draw[dashed] (cm1) to[out=180, in=90] (cb1);
    \draw (ca1) to[out=-90,in=0] node[pos=0.5, below right=0.05cm] {$1$} (dots1);
    \draw (cb1) to[out=-90, in=180] node[pos=0.5, below left=0.05cm] {$1$} (dots1);
    
    \end{tikzpicture}
    }
    \caption{Case 2, $\xi=0$.}
    \label{fig: second case}
    \end{subfigure}
\end{tabular}
\vspace{-5pt}
    \caption{Illustrations of Step 7.}
    \label{fig: step 7}
\end{figure}
\inline Step 8:
Generate the output $G$.
\\
Denote $E^\Delta = E^\Delta_1 \cup E^\Delta_2$. For each cycle $C_s \in \calC^\Delta$, the corresponding $s \in S^{\Delta}$ is connected to a dominator vertex with an edge $\hat{e}_s \in E^\Delta$ of weight one. Since there is one-to-one correspondence between $\calC^\Delta$ and $S^\Delta$, $E^\Delta$ covers all the vertices of $S^\Delta$, in the sense that for every $s \in S^\Delta$ exists the corresponding edge $\hat{e}_s \in E^\Delta$, such that $s \in \hat{e}_s$. 

After applying Steps 3--8, each edge in $G^\omega$ has weight either one or zero. On the other hand, each step preserves the weighted degrees in $G^\omega$, so $G^\omega$ is still a weighted realization of $d$. It follows that $G^\omega$ transforms into a simple graph $G$ with an edge $(i, j)$ whenever $\omega(i,j)=1$ in $G^\omega$. Clearly, $G$ realizes $d$.

Finally, since $S = S' \cup S^\Delta$, every nondominating vertex $s \in S$ is connecting to the dominating set by some edge in $E' \cup E^\Delta$. These edges have weight one in $G^\omega$, so $G$ contains them all. Therefore, $D$ is a dominating set in $G$.

\begin{lemma}
\label{lem: MDS equivalence}
There is a $\gamma$-prefix-dominated realization $\hat{G}$ of $(d,d)$ if and only if there is a $\gamma$-prefix-dominated realization $G$ of $d$.
\end{lemma}

\begin{proof}
$(\Rightarrow)$ 
If there $\gamma$-prefix-dominated realization of $(d,d)$, then the aforementioned algorithm produces a $\gamma$-prefix-dominated realization of $d$.

$(\Leftarrow)$ 
Let $G = (V, E)$ with $V = [1,n]$ be a realization of $d$, where $\deg(i) = d_i$ for $i \in [1,n]$, 
with a dominating set $D = [1,\gamma]$. 
Consider the 
graph $\hat{G} = (\hat{V}, \hat{W}, \hat{E})$ with $\hat{V} = \{v_1, \ldots, v_n\}$, $\hat{W} = \{w_1, \ldots, w_n\}$ and $(v_i, w_j) \in \hat{E}$ if and only if $(i, j) \in E$. 
Clearly, $\hat{G}$ realizes $(d,d)$ (satisfying property (D1)) and does not have edges $(v_i, w_i)$ for every $i \in [1,n]$ 
(satisfying 
(D2)). 
Let $\hat{D} = \hat{D}_V \cup \hat{D}_W$, where $\hat{D}_V = \{v_1, \ldots, v_{\gamma}\}$ and $\hat{D}_W= \{w_1, \ldots, w_{\gamma}\}$ are prefixes of $\hat{V}$ and $\hat{W}$, respectively. Since $D$ is a dominating set in $G$, every vertex in $\hat{W} \setminus \hat{D}_W$ is connected to some dominator vertex in $\hat{D}_V$ and every vertex in $\hat{V} \setminus \hat{D}_V$ is connected to some dominator vertex in $\hat{D}_W$. Thus, $\hat{D}$ is a dominating set in $\hat{G}$ (satisfying 
(D3)).
\end{proof}

\subsection{Reduction to Flow}
We have reduced the initial problem of finding a general realization with a dominating set of a certain type to the same problem in a bipartite setting.
Hence, we are left with the task of constructing, for a given sequence $d$ with $\domset(d)=\gamma$, a $\gamma$-prefix-dominated realization $\hat{G} = (V, W, \hat{E})$ for the sequence pair $(d,d)$.



To find the desired realization minimizing the dominating set, we go through every possible value of $\gamma$ 
and check if a suitable realizing graph $\hat{G}$ exists by reducing the problem to a flow problem. 
Specifically, given $d$ and a value $\gamma$, we construct a directed bipartite flow graph $G_{d,\gamma}$ with edge set $\calE$
as follows. 
Corresponding to $V$ and $W$, we have the sets $X=\{x_1,\ldots,x_n\}$ and $Y=\{y_1,\ldots,y_n\}$.
The node set of the flow graph is
\(
V ~=~ X ~\cup~ Y ~\cup~ X'_S ~\cup~ Y'_S ~\cup~ \{s, t\}
\),
where the nodes are organized into the following categories.
\begin{compactitem}
\item 
Candidate dominators: $X_D =\{x_i \mid i \in [1,\gamma]$ and $Y_D = \{y_j \mid j \in [1,\gamma]\}$,
\item 
Nondominating nodes: 
$X_S = \{x_i \mid i \in [\gamma+1,n]$ and $Y_S = \{y_j \mid j \in [\gamma+1,n]\}$,
\item 
Supplementary vertices: 
$X'_S = \{x'_i \mid i \in [\gamma+1,n]$ and $Y'_S = \{y'_j \mid j \in [\gamma+1,n]\}$.
\item 
Source and sink: \{s,t\}.
\end{compactitem}

The source $s$ is connected to the nodes of $X$. Similarly, the nodes of $Y$ are connected to the sink $t$. In addition, the supplementary nodes of $X'_S$ and $Y'_S$ are used to create the connections and flow capacities and enforce the degree and domination constraints. See Figure 
{\ref{fig:flow graph}}, in which the flow from $s$ to $t$ goes left-to-right.
%
%


The edges of $\calE$ are capacitated and directed (from left to right). We use the notation $(\alpha, \beta, \delta)$ for an edge that leads from the node $\alpha$ to the node $\beta$ and can carry up to $\delta$ units of flow. The source $s$ has edges leading to every node $x_i \in X$, with capacity $d_i$, to enforce the degree constraint for the corresponding vertex $v_i$ in the realizing graph $G$. Similarly, there are edges leading from every node $y_j \in Y$ to the sink $t$, with capacity $d_j$, to enforce the degree constraint for the vertices $w_j$ in $G$. 

Additional edges are employed to ensure the correctness of the reduction. 
We write $\tilde{X} \hourglass \tilde{Y}$ to indicate that the vertices $\tilde{X}$ and $\tilde{Y}$ induce an almost complete bipartite graph in $G_{d, \gamma}$, missing only the edges mentioned in property (D2) 
(i.e., $\{(x_i, y_i) \mid i \in [1,n]$ and $\{(x'_i, y'_i) \mid i \in [\gamma+1,n]$, 
which are not present in $G_{d,\gamma}$). The various groups of $V$ are connected by $\calE$ in the following way: 
$X_D \hourglass Y_D$, $X_D \hourglass Y_S$, $X_S \hourglass Y_D$ and $X'_S \hourglass Y'_S$ with each edge having capacity one. 
Finally, every $x_i \in X_S$ is connected to the corresponding $x'_i \in X'_S$ with capacity $d_i - 1$, and similarly, 
every $y'_j \in Y'_S$ is connected to the corresponding $y_j \in Y_S$ with capacity $d_j - 1$. Denote these connections by $X_S \equiv X'_S$ and $Y'_S \equiv Y_S$. The connections are depicted schematically in Figure 
\ref{fig:flow graph} and all the edges are listed below.
\begin{align*}
\calE ~=~ 
& \{(s, x_i, d_i) \mid i \in [1,n] ~\cup~ \{(y_j, t, d_j) \mid j \in [1,n]\} 
   ~\cup~ \{(x_i, y_j, 1) \mid i,j \in [1,\gamma], ~i \neq j\}  \\
& \cup~ \{(x_i, y_j, 1) \mid i \in [1,\gamma],~ j \in [\gamma+1,n]\} 
 ~\cup~ \{(x_i, y_j, 1) \mid i \in [\gamma+1,n],~ j \in [1,\gamma]\} \\
& \cup~ \{(x_i, x_i', d_i-1) \mid i \in [\gamma+1,n]\}  ~\cup~ \{(y'_j, y_j, d_j-1) \mid j \in [\gamma+1,n]\} \\
& \cup~ \{(x'_i, y'_j, 1) \mid i,j \in [\gamma+1,n],~ i \neq j\}
~.
\end{align*}

\begin{figure}[h!]
\centering
\resizebox{0.6\textwidth}{!}{
\tikzset{
myptr/.style={-{Stealth[scale=1.5]}},
}

\definecolor{r}{rgb}{1, 0, 0.0}
\definecolor{g}{rgb}{0.0, 0.5, 0.0}
\definecolor{b}{rgb}{0.0, 0, 1}
\definecolor{rg}{rgb}{1, 0.5, 0.0}
\definecolor{gb}{rgb}{0, 0.5, 1}
\def\thicc{1.5}

\begin{tikzpicture}[
    roundnode/.style={circle, draw=black, minimum size=35pt, inner sep=0pt, line width=1pt, font=\huge},
    squarebox/.style={draw=r, inner sep=5pt, font=\huge},
    roundedbox/.style={draw=g, rounded corners=15pt, inner sep=5pt, font=\huge}, 
    roundedboxprime/.style={draw=b, rounded corners=15pt, inner sep=5pt, font=\huge}, 
    dots/.style={
    draw=none, fill=none, font=\huge},
    edgenode/.style={font=\huge}
]

\node[dots] (x1) {$\vdots$};
\node[roundnode, below=1cm of x1] (xi) {$x_{i'}$};
\node[dots, below=1cm of xi] (xl) {$\vdots$};

\node[dots, below=1cm of xl] (xl1) {$\vdots$};
\node[roundnode, below=1cm of xl1] (xi') {$x_i$};
\node[dots, below=1cm of xi'] (xn) {$\vdots$};

\node[font = \huge, color=r, left=0.8cm of x1] (XD) {$X_D$};
\node[font = \huge, color=g, left=0.8cm of xn] (XS) {$X_S$};

\node[fit=(x1)(xn)] (xcenter) {};

\node[roundnode, left=4cm of xcenter] (s) {s};

\node[dots, below right =6cm of xl] (x'1) {$\vdots$};
\node[roundnode, below=1cm of x'1] (x'i) {$x'_i$};
\node[dots, below=1cm of x'i] (x'n) {$\vdots$};

\node[font = \huge, color=b, left=0.8cm of x'n] (X'S) {$X'_S$};

\node[dots, right=16cm of x1] (y1) {$\vdots$};
\node[roundnode, below=1cm of y1] (yj) {$y_{j'}$};
\node[dots, below=1cm of yj] (yk) {$\vdots$};

\node[dots, below=1cm of yk] (yk1) {$\vdots$};
\node[roundnode, below=1cm of yk1] (yj') {$y_j$};
\node[dots, below=1cm of yj'] (ym) {$\vdots$};

\node[font = \huge, color=r, right=0.8cm of y1] (YD) {$Y_D$};
\node[font = \huge, color=g, right=0.8cm of ym] (YS) {$Y_S$};

\node[fit=(y1)(ym)] (ycenter) {};

\node[roundnode, right=4cm of ycenter] (t) {t};

\node[dots, below left=6cm of yk] (y'1) {$\vdots$};
\node[roundnode, below=1cm of y'1] (y'j) {$y'_j$};
\node[dots, below=1cm of y'j] (y'm) {$\vdots$};

\node[font = \huge, color=b, right=0.8cm of y'm] (Y'S) {$Y'_S$};

\draw[line width=\thicc] (s) -- (xi) node[edgenode, pos=0.5, above=0.2cm] {$d_{i'}$};
\draw[line width=\thicc] (s) -- (xi') node[edgenode, pos=0.5, below=0.2cm] {$d_i$};

\draw[line width=\thicc] (yj) -- (t) node[edgenode, pos=0.75, above=0.2cm] {$d_{j'}$};
\draw[line width=\thicc] (yj') -- (t) node[edgenode, pos=0.75, below=0.2cm] {$d_j$};

\draw[line width=\thicc] (xi) -- (yj) node[edgenode, pos=0.5, above=0.2cm] {$1$};
\draw[line width=\thicc] (xi) -- (yj') node[edgenode, pos=0.5, above=0.2cm] {$1$};
\draw[line width=\thicc] (xi') -- (yj) node[edgenode, pos=0.5, above=0.2cm] {$1$};

\draw[line width=\thicc] (x'i) -- (y'j) node[edgenode, pos=0.5, above=0.2cm] {$1$};
\draw[line width=\thicc] (xi') -- (x'i) node[edgenode, pos=0.75, below left=0.1cm] {$d_i-1$};
\draw[line width=\thicc] (y'j) -- (yj') node[edgenode, pos=0.25, below right=0.1cm] {$d_j-1$};

\begin{scope}[on background layer]
    \node[squarebox, line width=\thicc, fit=(x1) (xi) (xl)] {};
    \node[roundedbox, line width=\thicc, fit=(xl1) (xi') (xn)] {};
    \node[roundedboxprime, line width=\thicc, fit=(x'1) (x'i) (x'n)] {};
    \node[squarebox, line width=\thicc, fit=(y1) (yj) (yk)] {};
    \node[roundedbox, line width=\thicc, fit=(yk1) (yj') (ym)] {};
    \node[roundedboxprime, line width=\thicc, fit=(y'1) (y'j) (y'm)] {};
\end{scope}

\end{tikzpicture}
}
\caption{The flow graph. 
}
\label{fig:flow graph}
\end{figure}
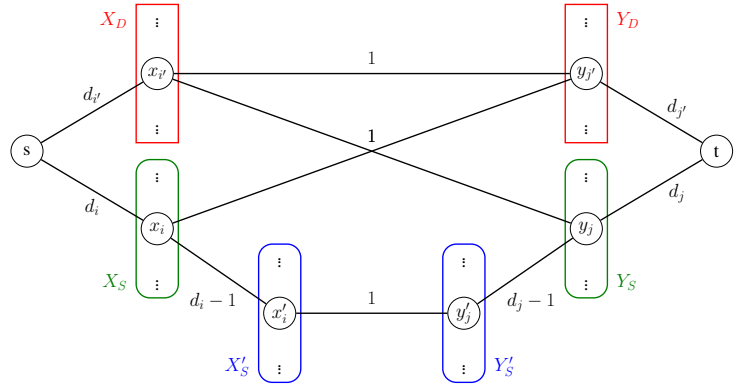

Next, we compute a maximum flow from the source $s$ to the sink $t$ in 
Since all the capacities are integral, there is an integral maximum flow, 
which can be computed with corresponding algorithms. 
Hence, we may assume that the computed maximum flow is integral.
Let $D = D_A \cup D_B$ with $D_A = \{v_{1}, \ldots, v_{\gamma}\}$ and 
$D_B =  \{w_{1}, \ldots, w_{\gamma}\}$. 
Note that by definition, $D_A = \{v_i \mid x_i \in X_D\}$ and $D_B = \{w_j \mid y_j \in Y_D\}$

\begin{lemma} 
\label{lem:proof of flow reduction}
There exists a bipartite graph $\hat{G} = (V, W, \hat{E})$ that realizes the degree sequence pair 
$(d, d)$ such that $D$ is a dominating set in $\hat{G}$ if and only if 
the maximum $s-t$ flow in $G_{d, \gamma}$ is $\sum_{i=1}^n d_i$.
\end{lemma}

\begin{proof}
$(\Leftarrow)$ Suppose the maximum $s$-$t$ flow in $G_{d,\gamma}$ is $\sum_{i=1}^n d_i$,
and let $\FLOW$ be an integral maximum flow.
Define the bipartite graph $\hat{G} = (V,W,\hat{E})$, where $V = \{v_1, \ldots, v_n\}$ and $W = \{w_1, \ldots, w_n\}$, 
by adding an edge $(v_i, w_j)$ to $\hat{E}$, for every $i, j \in [1,n]$ such that there is 
an edge $e$ in $G_{d,\gamma}$ of the form $(x_i, y_j)$ or $(x'_i, y'_j)$ with $\FLOW(e)=1$. 
Since the flow in integral, its values on the edges $(x_i, y_j)$ and $(x'_i, y'_j)$ are either $0$ or $1$.

To verify that $\hat{G}$ correctly realizes $(d,d)$, we first observe that the selected edges form a graph (rather than a multigraph). Indeed, for any pair $i, j \in [1,n]$, the flow graph $G_{d,\gamma}$ contains either only the edge $(x'_i, y'_j)$ (in case $x_i \in X_S$ and $y_j \in Y_S$) or only the edge $(x_i, y_j)$ (otherwise), but not both.

Next we show that $\deg_G(v_i)=d_i$ for every $i \in [1,n]$ and $\deg_G(w_j)=d_j$ for every $j \in [1,n]$. Note that the total flow from $s$ to $t$ is $\sum_{i=1}^n d_i$, which means that all edges from vertices $y_j$ to $t$ and from $s$ to $x_i$ are saturated.
Four cases should be considered. 
\begin{compactdesc}
\item[Case A:] 
If $v_i \in D_A$, then the degree $\deg_{\hat{G}}(v_i)$ of $v_i$ in $\hat{G}$ is set
      by the number of edges $(x_i, y_j)$ in $\calE$ with $\FLOW(x_i, y_j)=1$ for $j \in [1,n]$. 
      Since $x_i$ receives a flow of $d_i$ from the source $s$, by flow conservation it outputs 
      the same amount (in one unit flows) through edges $(x_i, y_j)$, hence its degree is $d_i$.
\item[Case B:]
If $w_j \in D_B$, then $\deg_{\hat{G}}(w_j)$ is determined by the number of edges $(x_i, y_j)$ in $\calE$ with $\FLOW(x_i, y_j)=1$ 
for $i \in [1,n]$, and its analysis is similar to Case A.
\item[Case C:]
 If $v_i \in V \setminus D_A$, then $\deg_{\hat{G}}(v_i)$ is determined by the number of edges $(x_i, y_j)$ and $(x'_i, y'_j)$ in $\calE$ carrying a flow unit for $j \in [1,n]$. We divide the output flow of $v_i$ into two terms: $\FLOW(s, x_i) = \sum_{j \in D_B} \FLOW(x_i, y_j) + \FLOW(x_i, x'_i)$. The first term increases the degree by the corresponding amount of flow, because each edge $(x_i, y_j)$ has capacity one. The rest of the flow goes to $x'_i$. Since $x'_i$ outputs flow through the edges $(x'_i, y'_j)$ with unit capacity, it increases the degree of $v_i$ by $\FLOW(x_i, x'_i)$. Hence, the degree of $v_i$ is  $\deg_{\hat{G}}(v_i) = \FLOW(s, x_i) = d_i$.
\item[Case D:]
If $w_j \in W \setminus D_B$, then $\deg_{\hat{G}}(w_j)$ is determined by the number of edges $(x_i, y_j)$ and $(x'_i, y'_j)$ in $\calE$ carrying a flow unit for $i \in [1,n]$, and its analysis is similar to Case C. 
\end{compactdesc}
Therefore, $\hat{G}$ correctly realizes the
degree sequence pair $(d, d)$, so property (D1) holds.
Also, since $G_{d,\gamma}$ contains neither $\{(x_i, y_i) \mid i \in [1,n] \}$ nor $\{(x'_i, y'_i) \mid i \in [\gamma+1,n]\}$, 
it follows that $\hat{G}$ does not have edges $\{(v_i, w_i) \mid i \in [1,n]\}$, hence it satisfies property (D2).

It remains to show property (D3), namely, that $D$ is a dominating set for $\hat{G}$. Since $y_j \in Y_S$ conserves flow and $\FLOW(y'_j, y_j) \le d_j - 1$, there is at least one $x_i \in X_D$ with $\FLOW(x_i, y_j) = 1$. Similarly, for every $x_i \in X_S$ there exists at least one $y_j \in Y_D$ such that $\FLOW(x_i, y_j) = 1$. It follows that at least one vertex from $D$ is connected to each $v_i \in V \setminus D$ and $w_j \in W \setminus D$.

\noindent
$(\Rightarrow)$
Conversely, suppose that $\hat{G} = (V, W, \hat{E})$, where $V = \{v_1, \ldots, v_n\}$ and $W = \{w_1, \ldots, w_n\}$, is a bipartite graph realizing $(d, d)$, such that $D$ is a dominating set. We need to define a flow function on $G_{d, \gamma}$ and prove that the maximum $s - t$ flow equals to $\sum_{i=1}^n d_i$. Obviously, the flow cannot exceed $\sum_{i=1}^n d_i$, so it suffices to prove that $\sum_{i=1}^n d_i$ is achievable. The flow is defined as follows.

\begin{enumerate}
\item 
\textbf{Saturating all edges from the source and to the sink}: \\
Set $\FLOW(s, x_i) \leftarrow d_i$
and $\FLOW(y_j, t) \leftarrow d_j$ ~~
for every $i,j \in [1,n]$.
\item 
\textbf{Flow to supplementary $X$ vertices}: \\
For every $i \in [\gamma+1,n]$, set $N_i \gets |\{(v_i, w_j) \in \hat{E} \mid j \in [\gamma+1,n]\}|$. \\
Set $\FLOW(x_i, x'_i) \leftarrow N_i \quad \text{for every } x_i \in X_S$
\\
(Note that since $v_i$ is connected with at least one dominating vertex, $N_i \le d_i - 1$, so the flows satisfy the capacity constraints.)
\item 
\textbf{Flow from supplementary $Y$ vertices}:
\\
For every $j \in [\gamma+1,n]$,
set $M_j \gets |\{(v_i, w_j) \in \hat{E} \mid i \in [\gamma+1,n]\}|$.
\\
Set $\FLOW(y'_j, y_j) \leftarrow M_j \quad \text{for every } y_j \in Y_S$.
\\
(Note that again the flows satisfy the capacity constraints.)
\item 
\textbf{Flow through edges}:
\\
For every pair $i, j \in [1,n]$ such that $(v_i, w_j) \in \hat{E}$:
\begin{compactenum}
\item If $x_i \in X_S$ and $y_j \in Y_S$, then  
set $\FLOW(x'_i, y'_j) \leftarrow 1$,
\item If either $x_i \in X_D$ or $y_j \in Y_D$, then 
set $\FLOW(x_i, y_j) \leftarrow 1$.
\end{compactenum}
\item 
\textbf{Completing the flow function}: 
\\
Set the flows of all other edges to $0$.
\end{enumerate}


\paragraph{Proof of Flow Legality.}
Note that by construction, the total flow out of the source $s$ and into the sink $t$ is $\sum_{i=1}^n d_i$.
It remains to verify that all other vertices conserve flow.
\begin{description}
\item[Node $x_i$:] 
By construction, the incoming flow at $x_i$ is $\FLOW(s, x_i) = d_i$. 
As for the outgoing flow, recall that 
there are exactly $d_i$ indices $j$ for which $(v_i,w_j)\in$ $\hat{E}$. 
\begin{compactitem}
\item If $i \in [1,\gamma]$, then the flows emanating from $x_i$ are $\FLOW(x_i, y_j)=1$ for precisely those indices. 
      Thus, the outgoing flow is $\sum_{j=1}^n \FLOW(x_i, y_j)=d_i$.
\item If $i \in [\gamma+1,n]$, then the outgoing flows from $x_i$ are $\FLOW(x_i, y_j)=1$ for $j \in [1,\gamma]$ 
      and $\FLOW(x_i, x'_i)= N_i = |\{ (v_i, w_j) \in \hat{E} \mid j \in [\gamma+1,n]|$. 
      Thus, the outgoing flow is $\sum_{j \in [1,\gamma]} \FLOW(x_i, y_j) + \FLOW(x_i, x'_i) = d_i$. 
\end{compactitem}
\item[Node $y_j$:]
By construction, the outgoing flow at $y_j$ is $\FLOW(y_j, t) = d_j$. 
As for the incoming flow, recall that there are exactly $d_j$ indices $i$ for which $(v_i,w_j)\in$ $\hat{E}$. 
\begin{compactitem}
    \item  If $j \in D_B$ then the flows incoming in $y_j$ are $\FLOW(x_i, y_j)=1$ for precisely those indices. Thus, the outgoing flow is $\sum_{i=1}^n \FLOW(x_i, y_j)=d_j$.
    \item If $j \notin D_B$ then the incoming flows to $y_j$ are $\FLOW(x_i, y_j)=1$ for $i \in D_A$ and $\FLOW(y'_j, y_j) = M_j = |\{ (v_i, w_j) \in \hat{E} \mid i \in V \setminus D_A \}|$. Thus, the outgoing flow is $\sum_{i \in D_A} \FLOW(x_i, y_j) + \FLOW(y'_j, y_j) = d_j$. 
\end{compactitem}
\item[Node $x'_i$:]
The incoming flow at $x'_i$ is $N_i$ by construction. The outgoing flow is $\sum_{ j \in [\gamma+1,n]} \FLOW(x'_i, y'_j) = N_i$ by definition. 
\item[Node $y'_j$:]
The outgoing flow from $y'_j$ is $M_j$ by construction. The incoming flow is $\sum_{i \in [\gamma+1,n]} \FLOW(x'_i, y'_j) = M_j$ by definition. 
\end{description}

Hence, the flow function defined above is valid, and the maximum $s-t$ flow in $G_{d, \gamma}$ equals to $\sum_{i=1}^n d_i$, implying the existence of a bipartite graph $G$ with $D$ as a dominating set.
\end{proof}

The next theorem states the main result of this section.

\begin{theorem}
\label{thm:real-MDS}
There exists a polynomial-time algorithm for constructing a realization of a given graphic sequence $d$ that also has a dominating set $D$ of minimum size (among all possible realizations of $d$).
\end{theorem}

In the next section we derive a more efficient algorithm by providing a succinct characterization for the problem and using it for a faster search.
\section{An Erdős--Gallai Type }

In this section we complement the previous result by providing an Erdős--Gallai Type characterization for the degree sequences $d$ that have a realization with a dominating set of size $\gamma$. 

It was proven in the previous section that a degree sequence $d$ has a realization with a minimum dominating set of size $\gamma$ if and only if the flow graph $G_{d, \gamma}$ admits a flow of size $\sum_{i=1}^n d_i$. 
Therefore, by the max--flow min--cut theorem, $d$ has a realization with a minimum dominating set of size $\gamma$ 
if and only if the capacity of every cut in $G_{d, \gamma}$ is at least $\sum_{i=1}^n d_i$.
This can serve as a basis for the sought characterization. Unfortunately, a naive implementation yields a characterization of length exponential in the input size. To overcome that difficulty, we apply the following strategy.
As a first step, we show that a minimum cut belongs to one of three families of cuts $\calF_1$, $\calF_2$, and $\calF_3$.
Each cut in the above families induces a constraint whose size is at least $\sum_{i=1}^n d_i$.
Next, we decrease the number of constraints by showing that it is sufficient to consider $O(n)$ constraints per family.

\subsection{The characterization}

We prove the following theorem.

\begin{theorem}
\label{thm:systems}
A sequence $d$ has a realization with a dominating set of size $\gamma$ if and only if the following systems of constraints are satisfied.
\begin{align}
\label{eqn:erdos-gallai 1}
\sum_{i=\gamma+1}^{\gamma+k} (d_i-1) 
& \leq k(k-1) - (n-\gamma) +\sum_{i=1}^\gamma d_i + \sum_{i=\gamma+k+1}^n \min\{k, d_i-1\}, ~ 
\mbox{\rm for every }
k \in [0,n-\gamma] \\
\label{eqn:erdos-gallai 2}
\sum_{i=1}^{k} d_i 
& \leq k(k-1) + \sum_{i=k+1}^n \min\{k, d_i\}, ~ 
\mbox{\rm for every }
k \in [0,n] \\
\label{eqn:erdos-gallai 3}
\sum_{i=\gamma+1}^{\gamma+k} d_i 
& \leq k(k-1) +\sum_{i=1}^\gamma \min\{k, d_i\} + \sum_{i=\gamma+k+1}^n \min\{k, d_i-1\}, ~ 
\mbox{\rm for every }
k \in [0,d_1]
\end{align}
\end{theorem}

Notice that System~\eqref{eqn:erdos-gallai 2} coincides with the Erd\H{o}s-Gallai conditions.
The rest of this subsection is devoted to proving the above theorem.

Let $c(S,T)$ be the capacity of an $s$-$t$ cut $(S,T)$ of $G_{d, \gamma}$.
Let $\calS$ be the set of $s$-$t$ cuts of the graph $G_{d, \gamma}$.
Clearly, $\calS = \calS_{00} \cup \calS_{01} \cup \calS_{10} \cup \calS_{11}$, where
\begin{align*}
\calS_{00} & = \set{(S,T) \mid S \cap X_D = \emptyset, T \cap Y_D = \emptyset} \\
\calS_{01} & = \set{(S,T) \mid S \cap X_D = \emptyset, T \cap Y_D \neq \emptyset} \\
\calS_{10} & = \set{(S,T) \mid S \cap X_D \neq \emptyset, T \cap Y_D = \emptyset} \\
\calS_{11} & = \set{(S,T) \mid S \cap X_D \neq \emptyset, T \cap Y_D \neq \emptyset} 
\end{align*}

We show that the capacities of minimum cuts in $\calS_{01}$ and in $\calS_{10}$ are the same.

\begin{lemma}
\label{lemma:symmetric}
$\min \set{c(S,T) \mid (S,T) \in \calS_{01}} = \min \set{c(S,T) \mid (S,T) \in \calS_{10}}$.
\end{lemma}
\begin{proof}
Let $(S, T) \in \calS_{01}$ and define $(S',T')$ as follows:
\begin{align*}
S' & = \{x_i \mid y_i \in T\} \cup  \{x'_i \mid y'_i \in T\} \cup \{y_i \mid x_i \in T\} \cup  \{y'_i \mid x'_i \in T\}, \\
T' & = \{x_i \mid y_i \in S\} \cup  \{x'_i \mid y'_i \in S\} \cup \{y_i \mid x_i \in S\} \cup  \{y'_i \mid x'_i \in S\}.
\end{align*}
Observe that $(S',T') \in \calS_{10}$ by construction.
Moreover, due to the symmetric structure of $G_{d, \gamma}$, there is a capacity preserving bijection 
between the edges in $(S,T)$ and the edges in $(S', T')$.
More specifically, $(x_i, y_j)$ maps to $(x_j, y_i)$, $(x'_i, y'_j)$ maps to $(x'_j, y'_i)$, 
$(x_i, x'_i)$ maps to $(y_i, y'_i)$, $(y_i, y'_i)$ maps to $(x_i, x'_i)$, $(s, x_i)$ maps to $(y_i, t)$,
and finally $(y_i, t)$ maps to $(s, x_i)$. Therefore, $c(S',T') = c(S,T)$. 
\end{proof}

Our next goal is to consider smaller cut families.
The following lemmas facilitate this goal.

\begin{lemma}
\label{lemma:modification emptyY}
Let $(S, T)$ be a cut of $G_{d, \gamma}$ such that $T \cap Y_D = \emptyset$.
Also, let $S' = S \cup X_S$ and $T' = T \setminus X_S$.
Then $c(S',T') \le c(S,T)$.
\end{lemma}
\begin{proof}
Let $x_i \in X_S \setminus S$.
By adding $x_i$ to $S'$ we remove the edge $(s,x_i)$ from the cut, but we may add 
the edge $(x_i,x'_i)$ to the cut (see \Cref{fig:flow graph}). Hence, 
\(
c(S,T) - c(S', T') \geq \sum_{x_i \in X_S \setminus S}(d_i - (d_i-1)) \ge 0
\).
\end{proof}

Symmetrical arguments imply that 

\begin{lemma}
\label{lemma:modification emptyX}
Let $(S,T)$ be a cut of $G_{d, \gamma}$ such that $S \cap X_D = \emptyset$.
Also, let $T' = T \cup Y_S$ and $S' = S \setminus Y_S$.
Then $c(S',T') \le c(S,T)$.
\end{lemma}

\begin{lemma}
\label{lemma:modification non-emptyY}
Let $(S, T)$ be a cut of $G_{d, \gamma}$ such that $T \cap Y_D \neq \emptyset$.
Also, let $S' = S \setminus(Z \cup Z')$ and $T'= T \cup (Z \cup Z')$, where 
$Z  = \{x_i \in S \mid x'_i \notin S, i \in [\gamma+1,n]\}$ and 
$Z' =  \{x'_i \in S\mid x_i \notin S, i \in [\gamma+1,n]\}$.
Then, $c(S',T') \le c(S,T)$.
\end{lemma}
\begin{proof}
Let $x_i \in Z$, where $i \in [\gamma+1,n]$.
The removal of $x_i$ adds the edge $(s,x_i)$ to the cut, but removes the edge $(x_i,x'_i)$ and at least 
one edge of the form $(x_i,y_j)$, where $j \in [1,\gamma]$.
Next let $x'_i \in Z'$, where $i \in [\gamma+1,n]$.
The removal of $x'_i$ removes edges of the form $(x'_i,y'_j)$,
where $y'_j \in Y'_S \cap T$, for $i \neq j$.
(See \Cref{fig:flow graph}.)
It follows that 
\[
c(S,T) - c(S', T') 
= \sum_{x_i \in Z} (d_i-1 + |T \cap Y_D| - d_i) + \sum_{x'_i \in Z'}  |T \cap Y'_S \setminus \{y'_i\}|
\geq |Z|(|T \cap Y_D| - 1) \ge 0
~. \qedhere
\]
\end{proof}

Symmetrical arguments imply that 

\begin{lemma}
\label{lemma:modification non-emptyX}
Let $(S, T)$ be a cut of $G_{d, \gamma}$ such that $S \cap X_D \neq \emptyset$.
Also, let $S' = S \cup (Z \cup Z')$ and $T'= T \setminus (Z \cup Z')$, where 
$Z  = \{y_i \mid y'_i \notin S, i \in [\gamma+1,n]$ and $Z' =  \{y'_i \mid y_i \notin S, i \in [\gamma+1,n]\}$.
Then, $c(S',T') \le c(S,T)$.
\end{lemma}

Define the following cut families:
\begin{itemize}
\item $\calF_1 = \set{(S^1_{I,J},V\setminus S^1_{I,J}) \mid I,J \subseteq [\gamma+1,n]}$, 

      where $S^1_{I,J} = \{x_i' \mid i\in I\} \cup \{y'_j \mid j \in J\} \cup X_S \cup Y_D$.
\item $\calF_2 = \{(S^2_{I,J}, V\setminus S^2_{I,J}) \mid I,J \subseteq [1,n]\}$, 

      where $S^2_{I,J} = \{x_i \mid i\in I\} \cup \{y_j \mid j \in J\} \cup
             \{x_i' \mid i \in I \cap [\gamma+1,n]\} \cup \{y'_j \mid j \in J \cap [\gamma+1,n]\}$.

\item $\calF_3 = \{(S^3_{I,J}, V\setminus S^3_{I,J}) \mid I \subseteq [\gamma+1,n], ~ J \subseteq [1,n]\}$,

    where $S^3_{I,J} = \{x_i, ~x'_i \mid i \in I\} \cup \{y_j \mid j \in J \cap [1,\gamma]\} 
    \cup \{y'_j \mid j \in J \cap [\gamma+1,n]\}$.

\end{itemize}

We obtain the following.

\begin{corollary}
We have that
\begin{itemize}
\item $\min \set{c(S,T) \mid (S,T) \in \calF_1} \leq \min \set{c(S,T) \mid (S,T) \in \calS_{00}}$.
\item $\min \set{c(S,T) \mid (S,T) \in \calF_2} \leq \min \set{c(S,T) \mid (S,T) \in \calS_{11}}$.
\item $\min \set{c(S,T) \mid (S,T) \in \calF_3} \leq \min \set{c(S,T) \mid (S,T) \in \calS_{10}}$.
\end{itemize}
\end{corollary}
\begin{proof}
Let $(S,T)$ be a minimum cut in $G_{\gamma, d}$. 
We consider the following three cases:
\begin{enumerate}[{Case} 1:]
\item If $(S,T) \in \calS_{00}$, then by \Cref{lemma:modification emptyY,lemma:modification emptyX}, 
      there exists a cut $(S',T') \in \calF_1$ such that $c(S', T') \leq c(S,T)$.

\item If $(S,T) \in \calS_{11}$, then by \Cref{lemma:modification non-emptyY,lemma:modification non-emptyX}, 
      there exists a cut $(S',T') \in \calF_2$ such that $c(S', T') \leq c(S,T)$.

\item If $(S,T) \in \calS_{10}$, then by \Cref{lemma:modification non-emptyY,lemma:modification emptyX}, 
      there exists a cut $(S',T') \in \calF_3$ such that $c(S', T') \leq c(S,T)$.
      \qedhere
\end{enumerate}
\end{proof}

The next lemma will helps us to reduce the number of constraints in each family.

\begin{lemma}
\label{lemma:system equivalence}
For a non-decreasing sequence of integers $\{a_i\}_{i=1}^n$ and 
a non-decreasing function $f: [0,m] \to \mathbb{R}$ the following two constraint systems are equivalent
\begin{align}
\sum_{i\in I} a_i 
& \le f(|I|) + \sum_{i\in I} \min\{|I|-1, a_i\} + \sum_{i \in [1,n] \setminus I} \min\{|I|, a_i\}, 
   ~ I \subseteq [1,n],~ |I| \le m 
\label{eqn:system1}
\\
\sum_{i = 1}^k a_i 
& \le f(k) + k(k-1) + \sum_{i = k+1}^n \min\{k, a_i\}, ~ k \in [0,m]
\label{eqn:system2}
\end{align}
\end{lemma}
\begin{proof}
First we prove that System~\eqref{eqn:system1} is equivalent to
\begin{align}
\label{eqn:system3}
\sum_{i \in I} a_i 
\le f(|I|) + |I|(|I|-1) + \sum_{i \in [1,n] \setminus I} \min\{|I|, a_i\}, ~ I \subseteq [1,n], ~|I| \le m ~. 
\end{align}
Since $\min\{|I|-1, a_i\} \le |I|-1$, it follows that System~\eqref{eqn:system1} implies System~\eqref{eqn:system3}. 
On the other hand, fix $I$ and consider $I' = \{i \in I \mid |I| -1\le a_i \}$. Clearly, $|I'| \le |I|$. 
The inequality for $I$ in System~\eqref{eqn:system1} can be rewritten as
\begin{align}
\label{eqn:system1-I}
\sum_{i\in I'} a_i\le f(|I|) + (|I|-1)|I'| +\sum_{i\in[1,n]\setminus I} \min\{|I|, a_i\}
~.
\end{align}
Now consider the inequality for $I'$ in System~\eqref{eqn:system3}:
\[
\sum_{i\in I'} a_i 
\leq f(|I'|) +  (|I'|-1)|I'|+\sum_{i\in[1,n]\setminus I'} \min\{|I'|, a_i\} ~.
\]
Since $f$ is non-decreasing we have that
\begin{align*}
\sum_{i\in I'} a_i 
& \leq f(|I|) +  (|I'|-1)|I'| + (|I|-|I'|)|I'| +\sum_{i\in [1,n] \setminus I} \min\{|I'|, a_i\} \\
& =    f(|I|) +  (|I|-1)|I'| + \sum_{i \in [1,n] \setminus I} \min\{|I|, a_i\}
~,
\end{align*}
namely \eqref{eqn:system1-I} is satisfied.

Clearly, System~\eqref{eqn:system2} is implied by System~\eqref{eqn:system3} 
by choosing $I = [1,k]$. On the other hand, let $k = |I|$. 
We have that 
\[
\sum_{i\in I} a_i 
~\leq~ \sum_{i=1}^{k} a_i 
~\leq~ f(k) + k(k-1) + \sum_{i = k+1}^n \min\{k, a_i\} 
~\leq~ f(|I|) + |I|(|I|-1) + \sum_{i \in [1,n] \setminus I} \min\{|I|, a_i\}
~,
\]
which means that System~\eqref{eqn:system3} follows from System~\eqref{eqn:system2}.
\end{proof}   

In the next three lemmas we use \Cref{lemma:system equivalence} to reduce the number of constraints in each family. Next figures follow the convention that yellow blocks belong to $S$ and blue blocks belong to $T$ in a $s$-$t$ cut $(S, T)$. Multicolored blocks may contain vertices from both parts of the cut.

\begin{lemma}
\label{lemma:F1}
$c(S,T) \geq \sum_{i=1}^k d_i$, for every $(S,T) \in \calF_1$ if and only if 
System~\eqref{eqn:erdos-gallai 1} is satisfied.
\end{lemma}
\begin{proof}
Consider a cut $(S^1_{I,J}, V\setminus S^1_{I,J}) \in \calF_1$, where $I, J \subseteq [\gamma+1,n]$, 
and recall that 
\[
S^1_{I,J} = \{x_i' \mid i\in I\} \cup \{y'_j \mid j \in J\} \cup X_S \cup Y_D
~.
\]
We have that
\begin{align*}
c(S^1_{I,J}, V\setminus S^1_{I,J}) 
& = \sum_{i=1}^\gamma d_i + \sum_{i\in[\gamma+1,n]\setminus I} (d_i-1) +
    \sum_{i \in I, j \in [\gamma+1,n] \setminus J, i \neq j} 1 +
    \sum_{j \in J} (d_j-1) + \sum_{j=1}^\gamma d_j \\
& = 2\sum_{i=1}^\gamma d_i + \sum_{i \in [\gamma+1,n] \setminus I} (d_i-1) +
    \sum_{j \in [\gamma+1,n] \setminus (I \cup J)} |I| + \sum_{j \in I \setminus J}(|I|-1) +
    \sum_{j \in J} (d_j-1) 
~.
\end{align*}
(See example in~\Cref{fig:flow graph erdos-gallai 1}.)
Hence the capacity of all cuts in $\calF_1$ is least $\sum_{i=1}^n d_i$ if and only if 
for every $I, J \subseteq [\gamma+1,n]$
\[
\sum_{i\in I}d_i 
\leq \sum_{i=1}^\gamma d_i - (n-\gamma -|I|) + 
     \sum_{j \in [\gamma+1,n] \setminus (I\cup J)} |I| + 
     \sum_{j \in I \setminus J}(|I|-1) + \sum_{j \in J} (d_j-1)
~.
\]
It is not hard to verify that this system is satisfied if and only if for every $I \subseteq [\gamma+1,n]$
\[
\sum_{i\in I} (d_i-1) 
\leq \sum_{i=1}^\gamma d_i - (n-\gamma) + 
     \sum_{j \in [\gamma+1,n] \setminus I} \min\{|I|, d_j-1\} + 
     \sum_{j \in I} \min\{|I|-1, d_j-1\}
~.
\]
The lemma follows due to \Cref{lemma:system equivalence} applied with $a_i = d_{\gamma+i}-1$, for $i \in [1,n-\gamma]$,
$f(|I|) = \sum_{i=1}^\gamma d_i - (n-\gamma)$ and $m=n$.
\end{proof}

\begin{figure}[t!]
\centering
\resizebox{0.6\textwidth}{!}{
\tikzset{
myptr/.style={-{Stealth[scale=1.5]}},
}

\definecolor{r}{rgb}{1, 0, 0.0}
\definecolor{g}{rgb}{0.0, 0.5, 0.0}
\definecolor{b}{rgb}{0.0, 0, 1}
\definecolor{rg}{rgb}{1, 0.5, 0.0}
\definecolor{gb}{rgb}{0, 0.5, 1}
\def\thicc{1.5}

\begin{tikzpicture}[
    roundnode/.style={circle, draw=black, minimum size=35pt, inner sep=0pt, line width=1pt, font=\huge},
    squarebox/.style={draw=black, inner sep=5pt, font=\huge},
    roundedbox/.style={draw=black, rounded corners=15pt, inner sep=5pt, font=\huge}, 
    roundedboxprime/.style={draw=black,rounded corners=15pt, inner sep=5pt, font=\huge}, 
    dots/.style={
    draw=none, fill=none, font=\huge},
    edgenode/.style={font=\huge}
]

\node[dots] (x1) {$\vdots$};
\node[roundnode, below=1.3cm of x1] (xi) {$x_{i'}$};
\node[dots, below=0.7cm of xi] (xl) {$\vdots$};

\node[dots, below=2cm of xl] (xl1) {$\vdots$};
\node[roundnode, below=1cm of xl1] (xi') {$x_i$};
\node[dots, below=1cm of xi'] (xn) {$\vdots$};

\node[font = \huge, left=0.8cm of x1] (XD) {$X_D$};
\node[font = \huge, above=0.5cm of xl1] (XS) {$X_S$};

\node[fit=(x1)(xn)] (xcenter) {};

\node[roundnode, fill=yellow!20, left=3cm of xcenter] (s) {s};

\node[dots, below right=4cm of xl1] (x'1) {$\vdots$};
\node[roundnode, below=0.7cm of x'1] (x'i) {$x'_i$};
\node[dots, below=1.3cm of x'i] (x'n) {$\vdots$};

\node[font = \huge, above=0.5cm of x'1] (X'S) {$X'_S$};

\node[dots, right=10cm of x1] (y1) {$\vdots$};
\node[roundnode, below=0.7cm of y1] (yj) {$y_{j'}$};
\node[dots, below=1.3cm of yj] (yk) {$\vdots$};

\node[dots, below=2cm of yk] (yk1) {$\vdots$};
\node[roundnode, below=1cm of yk1] (yj') {$y_j$};
\node[dots, below=1cm of yj'] (ym) {$\vdots$};

\node[font = \huge, right=0.8cm of y1] (YD) {$Y_D$};
\node[font = \huge, above=0.5cm of yk1] (YS) {$Y_S$};

\node[fit=(y1)(ym)] (ycenter) {};

\node[roundnode, fill=blue!20, right=3cm of ycenter] (t) {t};

\node[dots, below left=4cm of yk1] (y'1) {$\vdots$};
\node[roundnode, below=1.3cm of y'1] (y'j) {$y'_j$};
\node[dots, below=0.7cm of y'j] (y'm) {$\vdots$};

\node[font = \huge, above=0.5cm of y'1] (Y'S) {$Y'_S$};

\draw[line width=\thicc] (s) -- (xi) node[edgenode, pos=0.5, above=0.2cm] {$d_{i'}$};
\draw[line width=\thicc] (s) -- (xi') node[edgenode, pos=0.5, below=0.2cm] {$d_i$};

\draw[line width=\thicc] (yj) -- (t) node[edgenode, pos=0.75, above=0.2cm] {$d_{j'}$};
\draw[line width=\thicc] (yj') -- (t) node[edgenode, pos=0.75, below=0.2cm] {$d_j$};

\draw[line width=\thicc] (xi) -- (yj) node[edgenode, pos=0.5, above=0.2cm] {$1$};
\draw[line width=\thicc] (xi) -- (yj') node[edgenode, pos=0.5, above=0.2cm] {};
\draw[line width=\thicc] (xi') -- (yj) node[edgenode, pos=0.47, above=0.2cm] {$1$};

\draw[line width=\thicc] (x'i) -- (y'j) node[edgenode, pos=0.5, above=0.2cm] {$1$};
\draw[line width=\thicc] (xi') -- (x'i) node[edgenode, pos=0.9, below left=0.1cm] {$d_i-1$};
\draw[line width=\thicc] (y'j) -- (yj') node[edgenode, pos=0.2, below right=0.1cm] {$d_j-1$};

\begin{scope}[on background layer]
    \node[squarebox, fill=blue, fill opacity=0.2, line width=\thicc, fit=(x1) (xi) (xl)] {};
    \node[squarebox, fill=yellow, fill opacity=0.2, line width=\thicc, fit=(xl1) (xi') (xn)] {};
    
     \draw[black, fill=blue, fill opacity=0.2, line width=\thicc] ($(x'1.north west)+(-0.56,0.2)$)  rectangle ($(x'1.south east)+(0.56, 0.7)$);
     
     \draw[black, fill=yellow, fill opacity=0.2, line width=\thicc] ($(x'i.north west)+(-0.34,1.585)$)  rectangle ($(x'i.south east)+(0.34,-0.585)$);

     \draw[black, fill=blue, fill opacity=0.2, line width=\thicc] ($(x'n.north west)+(-0.56,0.9)$)  rectangle ($(x'n.south east)+(0.56,-0.2)$);
    
    \node[squarebox, fill=yellow, fill opacity=0.2, line width=\thicc, fit=(y1) (yj) (yk)] {};
    \node[squarebox, fill=blue, fill opacity=0.2, line width=\thicc, fit=(yk1) (yj') (ym)] {};

    \draw[black, fill=blue, fill opacity=0.2, line width=\thicc] ($(y'1.north west)+(-0.56,0.2)$)  rectangle ($(y'1.south east)+(0.56,-0.9)$);
     
    \draw[black, fill=yellow, fill opacity=0.2, line width=\thicc] ($(y'j.north west)+(-0.34,0.585)$)  rectangle ($(y'j.south east)+(0.34,-0.985)$);

    \draw[black, fill=blue, fill opacity=0.2, line width=\thicc] ($(y'm.north west)+(-0.56,-0.1)$)  rectangle ($(y'm.south east)+(0.56,-0.2)$);
\end{scope}

\end{tikzpicture}
}
\caption{$s$-$t$ cut from $\calF_1$ in the flow graph.}
\label{fig:flow graph erdos-gallai 1}
\end{figure}
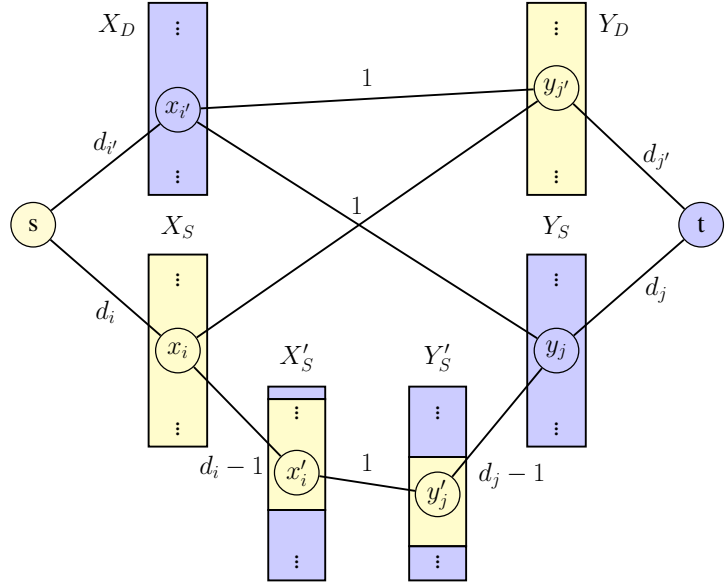

\begin{lemma}
\label{lemma:F2}
$c(S,T) \geq \sum_{i=1}^k d_i$, for every $(S,T) \in \calF_2$ if and only if 
System~\eqref{eqn:erdos-gallai 2} is satisfied.
\end{lemma}
\begin{proof}
Consider a cut $(S^2_{I,J}, V\setminus S^2_{I,J}) \in \calF_2$, where $I, J \subseteq[1,n]$, 
and recall that 
\[
S^2_{I,J} 
= \{x_i \mid i\in I\} \cup \{y_j \mid j \in J\} \cup
  \{x_i' \mid i \in I \cap [\gamma+1,n]\} \cup \{y'_j \mid j \in J \cap [\gamma+1,n]\}
~.
\]
Observe that
\begin{align*}
c(S^2_{I,J}, V\setminus S^2_{I,J}) 
& = \sum_{i \in [1,n] \setminus I} d_i + 
    \sum_{i \in I, j \in [1,n] \setminus J, i \neq j} 1 +
    \sum_{j \in J} d_j \\
& = \sum_{i \in [1,n] \setminus I} d_i + \sum_{j \in [1,n] \setminus (I\cup J)} |I| + 
    \sum_{j\in I \setminus J}(|I|-1) + \sum_{j \in J} d_j
~.
\end{align*}
(See example in~\Cref{fig:flow graph erdos-gallai 2}.)
It follows that the capacity of all cuts in $\calF_2$ is at least $\sum_{i=1}^n d_i$ if and only if 
for every $I, J \subseteq [1,n]$
\[
\sum_{i\in I} d_i \le \sum_{j \in [1,n] \setminus (I \cup J)} |I| + \sum_{j \in I \setminus J}(|I|-1) + \sum_{j \in J} d_j
~.
\]
It is not hard to verify that this system is satisfied if and only if for every $I \subseteq[1,n]$
\[
\sum_{i\in I} d_i \le \sum_{i \in [1,n] \setminus I} \min\{|I|, d_i\} + \sum_{i\in I}\min\{|I|-1, d_i\}
~.
\]
The lemma follows due to \Cref{lemma:system equivalence} applied with $a_i = d_i$, for $i \in [1,n]$,
$f(|I|) = 0$ and $m=n$.
\end{proof}

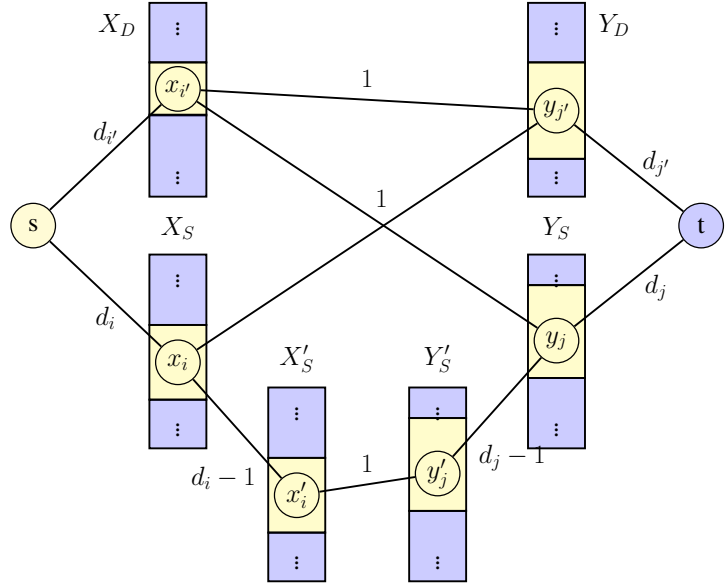
\begin{figure}[t!]
\centering
\resizebox{0.6\textwidth}{!}{
\tikzset{
myptr/.style={-{Stealth[scale=1.5]}},
}

\definecolor{r}{rgb}{1, 0, 0.0}
\definecolor{g}{rgb}{0.0, 0.5, 0.0}
\definecolor{b}{rgb}{0.0, 0, 1}
\definecolor{rg}{rgb}{1, 0.5, 0.0}
\definecolor{gb}{rgb}{0, 0.5, 1}
\def\thicc{1.5}

\begin{tikzpicture}[
    roundnode/.style={circle, draw=black, minimum size=35pt, inner sep=0pt, line width=1pt, font=\huge},
    squarebox/.style={draw=black, inner sep=5pt, font=\huge},
    roundedbox/.style={draw=black, rounded corners=15pt, inner sep=5pt, font=\huge}, 
    roundedboxprime/.style={draw=black,rounded corners=15pt, inner sep=5pt, font=\huge}, 
    dots/.style={
    draw=none, fill=none, font=\huge},
    edgenode/.style={font=\huge}
]

\node[dots] (x1) {$\vdots$};
\node[roundnode, below=0.7cm of x1] (xi) {$x_{i'}$};
\node[dots, below=1.3cm of xi] (xl) {$\vdots$};

\node[dots, below=2cm of xl] (xl1) {$\vdots$};
\node[roundnode, below=1.3cm of xl1] (xi') {$x_i$};
\node[dots, below=0.7cm of xi'] (xn) {$\vdots$};

\node[font = \huge, left=0.8cm of x1] (XD) {$X_D$};
\node[font = \huge, above=0.5cm of xl1] (XS) {$X_S$};

\node[fit=(x1)(xn)] (xcenter) {};

\node[roundnode, fill=yellow!20, left=3cm of xcenter] (s) {s};

\node[dots, below right=4cm of xl1] (x'1) {$\vdots$};
\node[roundnode, below=1.3cm of x'1] (x'i) {$x'_i$};
\node[dots, below=0.7cm of x'i] (x'n) {$\vdots$};

\node[font = \huge, above=0.5cm of x'1] (X'S) {$X'_S$};

\node[dots, right=10cm of x1] (y1) {$\vdots$};
\node[roundnode, below=1.3cm of y1] (yj) {$y_{j'}$};
\node[dots, below=0.7cm of yj] (yk) {$\vdots$};

\node[dots, below=2cm of yk] (yk1) {$\vdots$};
\node[roundnode, below=0.7cm of yk1] (yj') {$y_j$};
\node[dots, below=1.3cm of yj'] (ym) {$\vdots$};

\node[font = \huge, right=0.8cm of y1] (YD) {$Y_D$};
\node[font = \huge, above=0.5cm of yk1] (YS) {$Y_S$};

\node[fit=(y1)(ym)] (ycenter) {};

\node[roundnode, fill=blue!20, right=3cm of ycenter] (t) {t};

\node[dots, below left=4cm of yk1] (y'1) {$\vdots$};
\node[roundnode, below=0.7cm of y'1] (y'j) {$y'_j$};
\node[dots, below=1.3cm of y'j] (y'm) {$\vdots$};

\node[font = \huge, above=0.5cm of y'1] (Y'S) {$Y'_S$};

\draw[line width=\thicc] (s) -- (xi) node[edgenode, pos=0.5, above=0.2cm] {$d_{i'}$};
\draw[line width=\thicc] (s) -- (xi') node[edgenode, pos=0.5, below=0.2cm] {$d_i$};

\draw[line width=\thicc] (yj) -- (t) node[edgenode, pos=0.75, above=0.2cm] {$d_{j'}$};
\draw[line width=\thicc] (yj') -- (t) node[edgenode, pos=0.75, below=0.2cm] {$d_j$};

\draw[line width=\thicc] (xi) -- (yj) node[edgenode, pos=0.5, above=0.2cm] {$1$};
\draw[line width=\thicc] (xi) -- (yj') node[edgenode, pos=0.5, above right=0.1cm] {$1$};
\draw[line width=\thicc] (xi') -- (yj) node[edgenode, pos=0.5, above=0.2cm] {};

\draw[line width=\thicc] (x'i) -- (y'j) node[edgenode, pos=0.5, above=0.2cm] {$1$};
\draw[line width=\thicc] (xi') -- (x'i) node[edgenode, pos=0.8, below left=0.1cm] {$d_i-1$};
\draw[line width=\thicc] (y'j) -- (yj') node[edgenode, pos=0.2, below right=0.1cm] {$d_j-1$};

\begin{scope}[on background layer]
    \draw[black, fill=blue, fill opacity=0.2, line width=\thicc] ($(x1.north west)+(-0.56,0.2)$)  rectangle ($(x1.south east)+(0.56,-0.6)$);
     
    \draw[black, fill=yellow, fill opacity=0.2, line width=\thicc] ($(xi.north west)+(-0.34,0.285)$)  rectangle ($(xi.south east)+(0.34,-0.285)$);

    \draw[black, fill=blue, fill opacity=0.2, line width=\thicc] ($(xl.north west)+(-0.56,1.22)$)  rectangle ($(xl.south east)+(0.56,-0.2)$);

    \draw[black, fill=blue, fill opacity=0.2, line width=\thicc] ($(xl1.north west)+(-0.56,0.2)$)  rectangle ($(xl1.south east)+(0.56,-0.9)$);
     
    \draw[black, fill=yellow, fill opacity=0.2, line width=\thicc] ($(xi'.north west)+(-0.34,0.585)$)  rectangle ($(xi'.south east)+(0.34,-0.585)$);

    \draw[black, fill=blue, fill opacity=0.2, line width=\thicc] ($(xn.north west)+(-0.56,0.3)$)  rectangle ($(xn.south east)+(0.56,-0.2)$);

    \draw[black, fill=blue, fill opacity=0.2, line width=\thicc] ($(x'1.north west)+(-0.56,0.2)$)  rectangle ($(x'1.south east)+(0.56,-0.9)$);
     
    \draw[black, fill=yellow, fill opacity=0.2, line width=\thicc] ($(x'i.north west)+(-0.34,0.585)$)  rectangle ($(x'i.south east)+(0.34,-0.585)$);

    \draw[black, fill=blue, fill opacity=0.2, line width=\thicc] ($(x'n.north west)+(-0.56,0.3)$)  rectangle ($(x'n.south east)+(0.56,-0.2)$);

    \draw[black, fill=blue, fill opacity=0.2, line width=\thicc] ($(y1.north west)+(-0.56,0.2)$)  rectangle ($(y1.south east)+(0.56,-0.6)$);
     
    \draw[black, fill=yellow, fill opacity=0.2, line width=\thicc] ($(yj.north west)+(-0.34,0.885)$)  rectangle ($(yj.south east)+(0.34,-0.885)$);

    \draw[black, fill=blue, fill opacity=0.2, line width=\thicc] ($(yk.north west)+(-0.56,0)$)  rectangle ($(yk.south east)+(0.56,-0.2)$);
    
    \draw[black, fill=blue, fill opacity=0.2, line width=\thicc] ($(yk1.north west)+(-0.56,0.2)$)  rectangle ($(yk1.south east)+(0.56,0.2)$);
     
    \draw[black, fill=yellow, fill opacity=0.2, line width=\thicc] ($(yj'.north west)+(-0.34,1.085)$)  rectangle ($(yj'.south east)+(0.34,-0.585)$);

    \draw[black, fill=blue, fill opacity=0.2, line width=\thicc] ($(ym.north west)+(-0.56,0.9)$)  rectangle ($(ym.south east)+(0.56,-0.2)$);

    \draw[black, fill=blue, fill opacity=0.2, line width=\thicc] ($(y'1.north west)+(-0.56,0.2)$)  rectangle ($(y'1.south east)+(0.56,0.2)$);
     
    \draw[black, fill=yellow, fill opacity=0.2, line width=\thicc] ($(y'j.north west)+(-0.34,1.085)$)  rectangle ($(y'j.south east)+(0.34,-0.585)$);

    \draw[black, fill=blue, fill opacity=0.2, line width=\thicc] ($(y'm.north west)+(-0.56,0.9)$)  rectangle ($(y'm.south east)+(0.56,-0.2)$);
\end{scope}

\end{tikzpicture}
}
\caption{$s$-$t$ cut from $\calF_2$ in the flow graph.}
\label{fig:flow graph erdos-gallai 2}
\end{figure}

\begin{lemma}
\label{lemma:F3}
$c(S,T) \geq \sum_{i=1}^k d_i$, for every $(S,T) \in \calF_3$ if and only if
$\sum_{i=1}^{\gamma} d_i \ge n- \gamma$ and System~\eqref{eqn:erdos-gallai 3} is satisfied.
\end{lemma}
\begin{proof}
Consider a cut $(S^3_{I,J}, V\setminus S^3_{I,J}) \in \calF_3$, where $I \subseteq [\gamma+1,n]$ 
and $J \subseteq [1,n]$, and recall that 
\[
S^3_{I,J} = \{x_i, ~x'_i \mid i \in I\} \cup \{y_j \mid j \in J\cap [1,\gamma]\} \cup \{y'_j \mid j \in J \cap [\gamma+1,n]\}
~.
\]
Observe that
\begin{align*}
& c(S^3_{I,J}, V\setminus S^3_{I,J}) 
=  \sum_{i=1}^\gamma d_i +\sum_{i\in [\gamma+1,n] \setminus I} d_i + 
    \sum_{i \in I, j \in [1,n] \setminus J, i\neq j} 1 +
    \sum_{j\in [1,\gamma] \cap  J} d_j + \sum_{j \in [\gamma+1,n] \cap J} (d_j-1) \\
& =  \sum_{i=1}^\gamma d_i + \sum_{i \in [\gamma+1,n] \setminus I} \!\! d_i + 
     \sum_{j \in [1,n] \setminus (I\cup J)} \!\! |I| + 
     \sum_{j\in I \setminus J} (|I|-1) + \sum_{j \in [1,\gamma] \cap  J} \!\! d_j + 
     \sum_{j \in [\gamma+1,n] \cap J} \!\! (d_j-1) ~. 
\end{align*}
(See example in~\Cref{fig:flow graph erdos-gallai 3}.)
It follows that the capacity of all cuts in $\calF_3$ is at least $\sum_{i=1}^n d_i$ if and only if 
for every $I \subseteq[\gamma+1,n]$ 
and $J \subseteq [1,n]$
\[
\sum_{i\in I}d_i 
\le \sum_{j \in [1,n] \setminus (I\cup J)} |I| + \sum_{j \in I \setminus J}(|I|-1) + 
    \sum_{j \in [1,\gamma] \cap  J} d_j + \sum_{j \in [\gamma+1,n] \cap J} (d_j-1) 
    ~.
\]
It is not hard to verify that this system is satisfied if and only if for every 
$I \subseteq [\gamma+1,n]$
\[
\sum_{i\in I} d_i 
\leq \sum_{j \in [1,\gamma]} \min\{|I|, d_j\} + 
     \sum_{j \in [\gamma+1,n] \setminus I} \min\{|I|, d_j-1\} + \sum_{j\in I}\min\{|I|-1, d_j-1\}
     ~.
\]
If $|I| > d_1$, then the system transforms to 
\[
n-\gamma 
\leq \sum_{j \in [1,\gamma]} d_j + \sum_{j \in [\gamma+1,n] \setminus I} d_j
~.
\]
Notice that all the equations for $|I| > d_1$ follow from the single one $\sum_{i=1}^\gamma d_i\ge n-\gamma$
(which corresponds to the case where $I = [\gamma+1,n]$).
Otherwise, if $|I| \le d_1$, then rewrite the system  as 
\[
\sum_{i\in I} (d_i-1) 
\leq -|I| + \sum_{j \in [1,\gamma]} \min\{|I|, d_j\} + 
     \sum_{j \in [\gamma+1,n] \setminus I} \min\{|I|, d_j-1\} + \sum_{j \in I} \min\{|I|-1, d_j-1\}
~.
\]
Observe that  $f(|I|) = -|I|+ \sum_{j \in [1, \gamma]} \min\{|I|, d_j\}$ is 
a non-decreasing function of $|I|$ in the range (i.e.,$|I| \leq d_1$). 
The lemma follows due to \Cref{lemma:system equivalence} with $a_i = d_{\gamma+i}-1$ for $i \in [1, n-\gamma]$, and $m=d_1$.
\end{proof}

\begin{figure}[t!]
\centering
\resizebox{0.6\textwidth}{!}{
\tikzset{
myptr/.style={-{Stealth[scale=1.5]}},
}

\definecolor{r}{rgb}{1, 0, 0.0}
\definecolor{g}{rgb}{0.0, 0.5, 0.0}
\definecolor{b}{rgb}{0.0, 0, 1}
\definecolor{rg}{rgb}{1, 0.5, 0.0}
\definecolor{gb}{rgb}{0, 0.5, 1}
\def\thicc{1.5}

\begin{tikzpicture}[
    roundnode/.style={circle, draw=black, minimum size=35pt, inner sep=0pt, line width=1pt, font=\huge},
    squarebox/.style={draw=black, inner sep=5pt, font=\huge},
    roundedbox/.style={draw=black, rounded corners=15pt, inner sep=5pt, font=\huge}, 
    roundedboxprime/.style={draw=black,rounded corners=15pt, inner sep=5pt, font=\huge}, 
    dots/.style={
    draw=none, fill=none, font=\huge},
    edgenode/.style={font=\huge}
]

\node[dots] (x1) {$\vdots$};
\node[roundnode, below=1cm of x1] (xi) {$x_{i'}$};
\node[dots, below=1cm of xi] (xl) {$\vdots$};

\node[dots, below=2cm of xl] (xl1) {$\vdots$};
\node[roundnode, below=1.3cm of xl1] (xi') {$x_i$};
\node[dots, below=0.7cm of xi'] (xn) {$\vdots$};

\node[font = \huge, left=0.8cm of x1] (XD) {$X_D$};
\node[font = \huge, above=0.5cm of xl1] (XS) {$X_S$};

\node[fit=(x1)(xn)] (xcenter) {};

\node[roundnode, fill=yellow!20, left=3cm of xcenter] (s) {s};

\node[dots, below right=4cm of xl1] (x'1) {$\vdots$};
\node[roundnode, below=1.3cm of x'1] (x'i) {$x'_i$};
\node[dots, below=0.7cm of x'i] (x'n) {$\vdots$};

\node[font = \huge, above=0.5cm of x'1] (X'S) {$X'_S$};

\node[dots, right=10cm of x1] (y1) {$\vdots$};
\node[roundnode, below=1.3cm of y1] (yj) {$y_{j'}$};
\node[dots, below=0.7cm of yj] (yk) {$\vdots$};

\node[dots, below=2cm of yk] (yk1) {$\vdots$};
\node[roundnode, below=1cm of yk1] (yj') {$y_j$};
\node[dots, below=1cm of yj'] (ym) {$\vdots$};

\node[font = \huge, right=0.8cm of y1] (YD) {$Y_D$};
\node[font = \huge, above=0.5cm of yk1] (YS) {$Y_S$};

\node[fit=(y1)(ym)] (ycenter) {};

\node[roundnode, fill=blue!20, right=3cm of ycenter] (t) {t};

\node[dots, below left=4cm of yk1] (y'1) {$\vdots$};
\node[roundnode, below=0.7cm of y'1] (y'j) {$y'_j$};
\node[dots, below=1.3cm of y'j] (y'm) {$\vdots$};

\node[font = \huge, above=0.5cm of y'1] (Y'S) {$Y'_S$};

\draw[line width=\thicc] (s) -- (xi) node[edgenode, pos=0.5, above=0.2cm] {$d_{i'}$};
\draw[line width=\thicc] (s) -- (xi') node[edgenode, pos=0.5, below=0.2cm] {$d_i$};

\draw[line width=\thicc] (yj) -- (t) node[edgenode, pos=0.75, above=0.2cm] {$d_{j'}$};
\draw[line width=\thicc] (yj') -- (t) node[edgenode, pos=0.75, below=0.2cm] {$d_j$};

\draw[line width=\thicc] (xi) -- (yj) node[edgenode, pos=0.5, above=0.2cm] {$1$};
\draw[line width=\thicc] (xi) -- (yj') node[edgenode, pos=0.5, above =0.2cm] {};
\draw[line width=\thicc] (xi') -- (yj) node[edgenode, pos=0.52, above =0.3cm] {$1$};

\draw[line width=\thicc] (x'i) -- (y'j) node[edgenode, pos=0.5, above=0.2cm] {$1$};
\draw[line width=\thicc] (xi') -- (x'i) node[edgenode, pos=0.8, below left=0.1cm] {$d_i-1$};
\draw[line width=\thicc] (y'j) -- (yj') node[edgenode, pos=0.1, below right=0.1cm] {$d_j-1$};

\begin{scope}[on background layer]
   \node[squarebox, fill=blue, fill opacity=0.2, line width=\thicc, fit=(x1) (xi) (xl)] {};

    \draw[black, fill=blue, fill opacity=0.2, line width=\thicc] ($(xl1.north west)+(-0.56,0.5)$)  rectangle ($(xl1.south east)+(0.56,-0.9)$);
     
    \draw[black, fill=yellow, fill opacity=0.2, line width=\thicc] ($(xi'.north west)+(-0.34,0.585)$)  rectangle ($(xi'.south east)+(0.34,-1.085)$);

    \draw[black, fill=blue, fill opacity=0.2, line width=\thicc] ($(xn.north west)+(-0.56,-0.2)$)  rectangle ($(xn.south east)+(0.56,-0.2)$);

    \draw[black, fill=blue, fill opacity=0.2, line width=\thicc] ($(x'1.north west)+(-0.56,0.2)$)  rectangle ($(x'1.south east)+(0.56,-0.9)$);
     
    \draw[black, fill=yellow, fill opacity=0.2, line width=\thicc] ($(x'i.north west)+(-0.34,0.585)$)  rectangle ($(x'i.south east)+(0.34,-1.085)$);

    \draw[black, fill=blue, fill opacity=0.2, line width=\thicc] ($(x'n.north west)+(-0.56,-0.2)$)  rectangle ($(x'n.south east)+(0.56,-0.2)$);

    \draw[black, fill=blue, fill opacity=0.2, line width=\thicc] ($(y1.north west)+(-0.56,0.2)$)  rectangle ($(y1.south east)+(0.56,-0.9)$);
     
    \draw[black, fill=yellow, fill opacity=0.2, line width=\thicc] ($(yj.north west)+(-0.34,0.585)$)  rectangle ($(yj.south east)+(0.34,-0.285)$);

    \draw[black, fill=blue, fill opacity=0.2, line width=\thicc] ($(yk.north west)+(-0.56,0.6)$)  rectangle ($(yk.south east)+(0.56,-0.2)$);
    
    \node[squarebox, fill=blue, fill opacity=0.2, line width=\thicc, fit=(yk1) (yj') (ym)] {};

    \draw[black, fill=blue, fill opacity=0.2, line width=\thicc] ($(y'1.north west)+(-0.56,0.2)$)  rectangle ($(y'1.south east)+(0.56,0.1)$);
     
    \draw[black, fill=yellow, fill opacity=0.2, line width=\thicc] ($(y'j.north west)+(-0.34,0.985)$)  rectangle ($(y'j.south east)+(0.34,-0.585)$);

    \draw[black, fill=blue, fill opacity=0.2, line width=\thicc] ($(y'm.north west)+(-0.56,0.9)$)  rectangle ($(y'm.south east)+(0.56,-0.2)$);
\end{scope}

\end{tikzpicture}
}
\caption{$s$-$t$ cut from $\calF_3$ in the flow graph.}
\label{fig:flow graph erdos-gallai 3}
\end{figure}
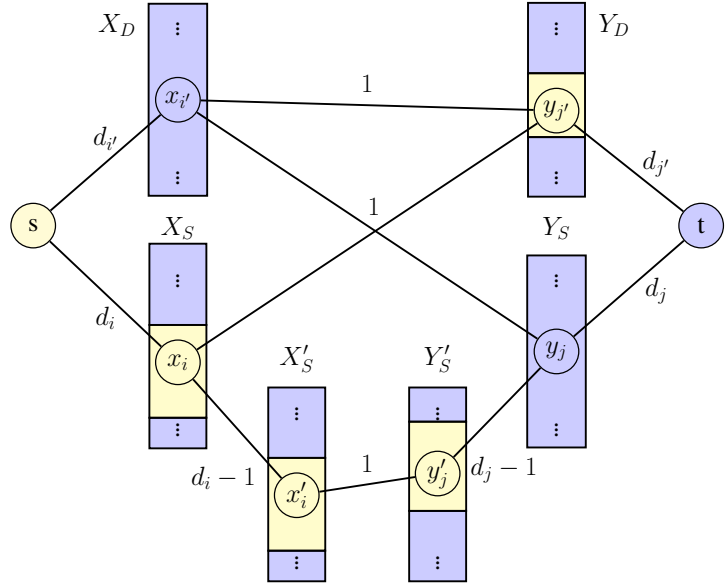

Since the equation for $k=0$ in System~\eqref{eqn:erdos-gallai 1} implies $\sum_{i=1}^\gamma d_i \geq n-\gamma$, 
\Cref{thm:systems} is implied by \Cref{lemma:F1,lemma:F2,lemma:F3}.

\subsection{Complexity analysis and a faster realization algorithm}

This section discusses the running time of two algorithms:
\begin{inparaenum}[(i)]
\item an algorithm to compute $\domset(d)$, which is implied by the characterizations given in \Cref{thm:systems}, and
\item our realization algorithm for \domsetDR.
\end{inparaenum}

Before addressing the running time of an algorithm for computing $\domset(d)$, 
consider first a decision algorithm.
Since the three systems in \Cref{thm:systems} contain $O(n)$ constraints, 
it follows that a naive implementation of a decision algorithm would result in 
a $O(n^2)$ running time. However, it is well known that the Erd\H{o}s-Gallai 
characterization (i.e., System~\eqref{eqn:erdos-gallai 2}) can be computed in $O(n)$ time, 
since the $k$th constraint can be computed from the $(k-1)$th constraint in $O(1)$ time.
A similar idea can be used for the other two systems. 
Hence, the decision algorithm can be implemented with a
$O(n)$ running time. Finally, Using a binary search an algorithm to compute 
$\domset(d)$ whose running time is $O(n \log n)$ is obtained.

As for the realization algorithm, 
given  $\gamma =\domset(d)$, to find a $\gamma$-prefix dominated realization of $d$ we first build the flow graph $G_{d, \gamma}$ and find a maximum $s$-$t$ flow in it. 
Since $G_{d, \gamma}$ has $O(n)$ vertices and $O(n^2)$ edges, 
a maximum flow can be computed in $O(n^3)$ time using the maximum flow
algorithm by Orlin~\cite{Orlin13}.
Given a maximum flow, adjacency matrix for $\gamma$-prefix-dominated realization $\hat{G}$ for the sequence pair $(d,d)$ can be computed in $O(n^2)$ time. 
Next we go through all the steps of transformation of $\hat{G}$ into $\gamma$-prefix dominated realization $G$ for $d$ and show that each step takes $O(n^3)$ time.
We need the following observations.

\begin{observation}
\label{obs: 2-dom path condition}
If $s \in S$ is adjacent to at least three vertices in $D$ in $\Ghalf$, then there is a 2-dom path $P[a,s,b,c]$. 
\end{observation}
\begin{proof}
Let $a, a', b \in D$ be three dominator vertices adjacent to $s$. Since $\Ghalf$ is an even graph, $b$ must be adjacent to some vertex $x \in \Vhalf$ besides $s$. If $x \neq a$, then $P[a,s,b,x]$ is 2-dom path. Otherwise, $P[a',s,b,a]$ is 2-dom path. 
\end{proof}

\begin{observation}
\label{obs: no new 2-dom paths}
For a vertex $s \in S$ if there is no 2-dom paths $P[a,s,b,c]$ in $\Ghalf$, then there will not be any such $2$-dom paths after applying MR1, MR2 or MR3 to any 2-dom path in $\Ghalf$. 
\end{observation}
\begin{proof}
By \Cref{obs: 2-dom path condition} we can assume that $s$ is connected to at most two dominator vertices in $\Ghalf$. Note that neither of the modifications can increase the number of dominator vertices adjacent to $s$ in $\Ghalf$. 
Therefore, if $s$ is connected to at most one dominator vertex, the Observation follows. 
    
Assume that $s$ is connected to exactly two dominator vertices denoted $a$ and $b$. 
Note that if $a, b$ or $s$ participated in any 2-dom path, it would imply 
the existence of $c \notin\{a,s,b\}$ adjacent to $a$ or to $b$ in $\Ghalf$. 
But then there exists a 2-dom path $P[a,s,b,c]$ or $P[b,s,a,c]$. It follows that modifications of 2-dom paths cannot create any new edges in $\Ghalf$ going through $a,b$ or $s$ and therefore cannot create any new 2-dom paths containing $s$. The observation follows.
\end{proof}

\begin{observation}
\label{obs: time to find 2-dom path}
Given $s \in S$ it is possible to find a 2-dom path $P[a,s,b,c]$ in $\Ghalf$ or make sure there is no such paths in $O(n)$ time.
\end{observation}
\begin{proof}
First, find the dominator vertices adjacent to $s$ in $O(n)$ time. 
If there is at most one such dominator vertex, then there is no 2-dom path
$P[a,s,b,c]$ in $\Ghalf$. If there is at least three dominator vertices, 
it is possible to find 2-dom paths in $O(n)$ time following the proof of 
\Cref{obs: 2-dom path condition}.

The remaining case is when $s$ is connected to exactly two dominator vertices $a$ and $b$. 
If $a$ and $b$ are connected only to each other and $s$ in the graph $\Ghalf$, 
then there is no 2-dom path going through $s$. 
Otherwise, if $a$ or $b$ is adjacent to $x \notin \{a,b,s\}$ in the graph $\Ghalf$, 
then the required 2-dom path is easily obtained.
Also, notice that the neighbors of $a$ and $b$ can be examined in $O(n)$ time.
The observation follows.
\end{proof}

\begin{enumerate}[{Step} 1:]
\item The adjacency matrix of the weighted graph $G^\omega$ can be computed from 
      the adjacency matrix of $\hat{G}$ in $O(n^2)$.

\item It takes $O(n^2)$ to compute the adjacency matrix of $\Ghalf$ using 
      the adjacency matrix of $G^\omega$.

\item Examine the vertices in $S$ one by one. 
      For $s \in S$ find a 2-dom path $P[a,s,b,c]$ in $\Ghalf$ and apply a corresponding modification to it. Since any modification decreases the number of neighbors of $s$ 
      in $\Ghalf$ by $2$, it follows that after $O(n)$ such steps no 2-dom paths $P[a,s,b,c]$ 
      remain in $\Ghalf$. 
      By \Cref{obs: no new 2-dom paths} and \ref{obs: time to find 2-dom path}, 
      all the 2-dom paths can be removed this way in $O(n^3)$ time.

\item $S'$, $E'$, $S^\Delta$ and $\calC^\Delta$ can be found straightforwardly
      in $O(n^2)$ time. Note that it takes only $O(1)$ time to check 
      if a vertex is a dominator or not, since $D = [1, \gamma]$.

\item The adjacency matrix of $H$ can be computed in $O(n^2)$ time 
      from $\calC^\Delta$ and the adjacency matrix of $\Ghalf$. 
      Since an Eulerian cycle can be found by linear time in the number of edges, 
      a partition of $H$ into disjoint cycles $\calC'$ takes $O(n^2)$ time.
    
\item Separating even cycles in $\calC$ and applying the corresponding modification 
      to them takes $O(n^2)$ time.
      
\item Since the vertices are enumerated, one can determine the intersection of the cycles 
      $C$ and $C'$ 
      in $O(|C| \log{|C|} + |C'| \log{|C'|})$ time by sorting their vertices 
      and then computing the intersection. If they intersect, it takes 
      $O(|C| +|C'|)$ time to apply the corresponding modification.
      If they do not intersect, the proof of \Cref{lem: right x and y} implies that 
      suitable $x$ and $y$ can be found in $O(n)$ time and $O(|C| +|C'|)$ time is needed to apply the modification. The complexity for all pairs is $O(n^3)$. 
      
\item At this step $G^\omega$ does not have edges of weight $1/2$, 
      so it can be transformed into $G$ in $O(n^2)$ time.
\end{enumerate}
It follows that the algorithm can be implemented in $O(n^3)$ time.

We conclude the section with the following result.

\begin{theorem}
\label{thm:real-MDS-fast}
There exists an $O(n^3)$ time algorithm for constructing a realization of a 
given graphic sequence $d$ that also has a dominating set $D$ of minimum size 
(among all possible realizations of $d$).
\end{theorem}

\section{Realization with Maximum Matching}
\label{s:MM-DR alg}

This section presents an algorithm for \matchDR, based on the Inverted Prefix Lemma \ref{lem: MM prefix} and a modification of the FHM realization algorithm
\footnote{An alternative algorithm 
the Inverted Prefix Lemma and the techniques of~\cite{Kundu73}. 
Here we present the FHM-based solution for uniformity of presentation.}. 
The algorithm follows the same two stages of the algorithm for \domsetDR described in \Cref{s:MDS-DR alg}, 
first reducing the problem to the bipartite setting and then reducing it to a maximum flow problem. 
Finally, the Inverted Prefix Lemma for maximum matching narrows the search to a polynomial number of cases.

\subsection{Prefix Lemma}
For \matchDR, 
there is a general prefix lemma 
by Gould, Jacobson and Lehel~{\cite{GJL99}}.

\begin{lemma}
{\bf (Arbitrary Prefix Lemma for \match) \cite{GJL99}}
\label{cl: Gould for MM}
If a sequence $d$ has a realization with a matching of size $\nu$, then $d$ has a realization with a matching of the same size such that $2 \nu$ vertices participating in the matching have the highest $2\nu$ degrees in $d$, i.e., $d_1, d_2, \ldots, d_{2\nu}$. 
\end{lemma}

For our purposes, however, we need a stronger type of prefix lemma, similar to the one used in~\cite{BockRautenbach19} for solving MM-DR over bipartite graphs. Towards proving this stronger claim, we introduce some notation. 
Consider a realization $G = (V, E)$ of a degree sequence $d$, 
with four vertices $x, y, u, v \in V$ such that $(x,u), (y,v) \in E$ and $(x,y), (u, v) \notin E$. 
Then the \emph{FLIP operation} transforms $G$ into a graph $G' = (V,E')$ by replacing the former 
two edges with the latter two, i.e., setting $E'=E \cup \{(x, y), (u, v)\} \setminus \{(x, u), (y, v)\}$. 
This operation preserves the degrees of individual vertices, but could lead to non-isomorphic realization. 
It was widely studied and used in different contexts, see~\cite{Fulkerson1960,EKM13,Barrus16}.

\begin{claim}
\label{cl: MM structure}
Let $G = (V, E)$ be a realization of $d$ with $V = \{v_1, v_2,\ldots, v_n\}$, 
such that $\deg_G(v_i) = d_i$ for every $i \in [1,n]$. 
If there is a matching $M  \subseteq E$ on the first $2\nu$ vertices $v_1,\ldots,v_{2\nu}$, 
then there is a realization $G' = (V, E')$ of $d$ with the "inverted" matching $M^{inv} = \{(v_i, v_{2\nu-i+1}) \mid i \in [1,\nu]\}$, $M^{inv} \subseteq E'$, 
and the same degrees $\deg_{G'}(v_i) = d_i$ for every $i \in [1,n]$.
\end{claim}
\begin{proof}
Call the pair $(G,M)$ \emph{valid} if $G=(V,E)$ is defined on $V = \{v_1, \ldots, v_n\}$ with $\deg(v_i) = d_i$, $M\subseteq E$, and $M$ is on $\{v_1, v_2, \ldots, v_{2\nu}\}$. 
Consider a valid pair $(G, M)$ and 
let $I(M)=\{ i \mid (v_i, v_{2\nu-i+1}) \notin M, ~i \in [1,\nu]$ and
$f(M) = \min I(M)$. 
If $I(M)=\emptyset$, then $M=M^{inv}$, and we are done.
    
We describe procedure \InvMatch that given a valid pair $(G,M)$ violating the claim produces a new valid pair $(G', M')$, such that either $(G',M')$ satisfies the claim or $f(M') > f(M)$. Repeatedly applying \InvMatch
gradually modifies $(G,M)$ while preserving validity and strictly increasing $f(M)$. Since $f$ is bounded above by $\nu$, the process terminates in at most $\nu$ steps and the obtained pair $(G', M')$ satisfies the claim. 

Given a pair $(G,M)$ not satisfying the claim, let $i = f(M)$ and $j = 2\nu -i +1$. 
Since $(v_i,v_j) \notin M$, there are edges $(v_i,v_x), (v_y,v_j) \in M$, and the involved vertices satisfy $\deg_G(v_i) \ge \deg_G(v_y)$ and $\deg_G(v_j) \le \deg_G(v_x)$ by the definition of $f$. 
Let $M' = M \cup \{(v_i,v_j), (v_x, v_y)\}\setminus \{(v_i,v_x), (v_y,v_j)\}$. 
If there is a realization $G'$ such that $M' \subseteq G'$, then the pair $(G', M')$ is valid and either satisfies the claim or increases $f$. Proceed according to one of the following cases to obtain $G'$. See \Cref{fig: MM prefix cases}.
\begin{enumerate}[{Case} 1:]
\item 
$(v_i, v_j), (v_x, v_y) \in E(G)$. Then simply take $G' = G$.

\item $(v_i, v_j), (v_x, v_y) \notin E(G)$. 
Then perform a FLIP operation on $G$ replacing $(v_i,v_x), (v_y,v_j)$ with $(v_i, v_j), (v_x, v_y)$. $G'$ obtained in this way satisfies $M' \subseteq G'$.

\item $(v_i, v_j) \in E(G)$ and $(v_x, v_y) \notin E(G)$. 
Then there is a vertex $v_z \in V(G) \setminus \{v_i.v_j,v_x,v_y\}$, such that $(v_x, v_z) \in E(G)$ and $(v_j, v_z) \notin E(G)$. 
Indeed, $\deg_G(v_x) \ge \deg_G(v_j)$ and $v_y$ is connected to $v_j$, but not to $v_x$, so there exists $v_z$ compensating for this. 
Obtain $G'$ by doing a FLIP replacing $(v_x, v_z), (v_j,v_y)$ with $(v_x,v_y), (v_j, v_z)$.

\item $(v_i, v_j) \notin E(G)$ and $(v_x, v_y) \in E(G)$.
Then there is a vertex $v_z \in V(G) \setminus \{v_i.v_j,v_x,v_y\}$, such that $(v_i, v_z) \in E(G)$ and $(v_y, v_z) \notin E(G)$. 
Similarly to the previous case, this is because $\deg_G(v_i) \ge \deg_G(v_y)$. 
Obtain $G'$ by doing a FLIP replacing $(v_i, v_z), (v_j,v_y)$ with $(v_i,v_j), (v_y, v_z)$.
\qedhere

\end{enumerate}
\end{proof}

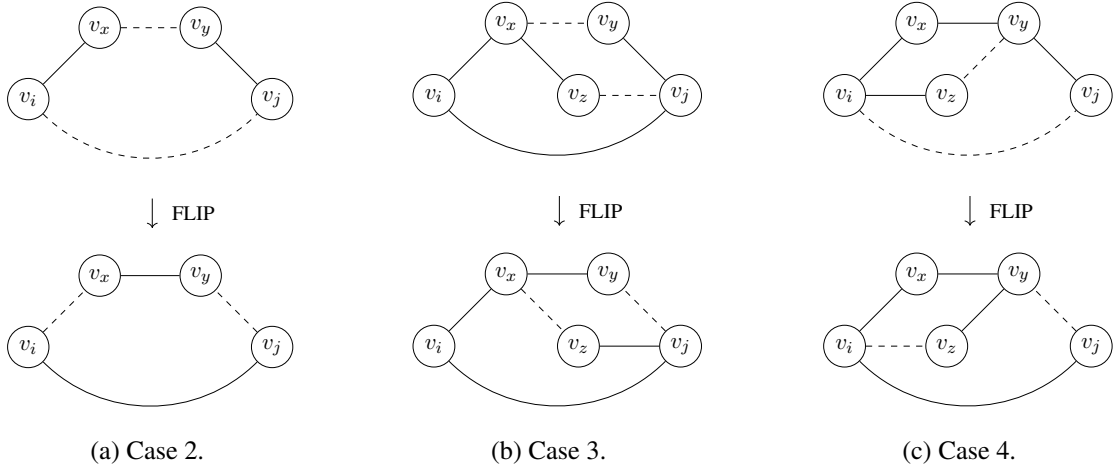
\begin{figure}[h!]
\centering
\begin{tabular}{c@{\hspace{5pt}}c@{\hspace{5pt}}c}
\begin{subfigure}{.305\textwidth}
    \centering
    \resizebox{0.8\textwidth}{!}{
    \tikzset{
    myptr/.style={-{Stealth[scale=1.5]}},
    }
    \begin{tikzpicture}[
        mynode/.style={circle, draw=black, minimum size=20pt, inner sep=0pt}
    ]
    
    \node[mynode] (i) {$v_i$};
    \node[mynode, above right = 1cm of i] (x) {$v_x$};
    \node[mynode, right =1cm of x] (y) {$v_y$};
    \node[mynode, below right = 1cm of y] (j) {$v_j$};
    \node[right = 0.4cm of x] (sup) {};
    
    \draw[dashed] (x) -- (y);
    \draw[dashed] (i) to[out=-45, in=-135] (j);
    \draw (i) -- (x);
    \draw (j) -- (y) ;
    
    \node[below=2.5cm of sup] (m1) {};
    \node[below=0.5cm of m1] (m2) {};
    \draw[->]        (m1)   -- (m2) node[pos=0.5, right=0.2cm] {\small FLIP};
    
    \node[mynode, below = 3.5cm of i] (i1) {$v_i$};
    \node[mynode, above right = 1cm of i1] (x1) {$v_x$};
    \node[mynode, right =1cm of x1] (y1) {$v_y$};
    \node[mynode, below right = 1cm of y1] (j1) {$v_j$};
    
    \draw (x1) -- (y1);
    \draw (i1)  to[out=-45, in=-135] (j1);
    \draw[dashed] (i1) -- (x1);
    \draw[dashed] (j1) -- (y1);
    
    \end{tikzpicture}
    }
    \caption{Case 2.}
    \label{fig: MM prefix 2}
\end{subfigure}
&
\begin{subfigure}{.32\textwidth}
    \centering
    \resizebox{0.8\textwidth}{!}{
    \tikzset{
    myptr/.style={-{Stealth[scale=1.5]}},
    }
    
    \definecolor{b}{rgb}{0.0, 0, 1}
    \definecolor{gr}{rgb}{0.7, 0.7, 0.7}
    
    \begin{tikzpicture}[
        mynode/.style={circle, draw=black, minimum size=20pt, inner sep=0pt}
    ]

    \node[mynode] (i) {$v_i$};
    \node[mynode, above right = 1cm of i] (x) {$v_x$};
    \node[mynode, right =1cm of x] (y) {$v_y$};
    \node[mynode, below right = 1cm of y] (j) {$v_j$};
    \node[below = 0.2cm of x] (sup) {};
    \node[mynode, below right = 1cm of x] (z) {$v_z$};
    \node[right = 0.4cm of x] (sup) {};
    
    \draw[dashed] (x) -- (y);
    \draw (i) to[out=-45,in=-135] (j);
    \draw (i) -- (x);
    \draw (j) -- (y);
    \draw (x) -- (z);
    \draw[dashed] (z) -- (j);

    \node[below=2.5cm of sup] (m1) {};
    \node[below=0.5cm of m1] (m2) {};
    \draw[->]  (m1)   -- (m2) node[pos=0.5, right=0.2cm] {\small FLIP};
    
    \node[mynode, below = 3.5cm of i] (i1) {$v_i$};
    \node[mynode, above right = 1cm of i1] (x1) {$v_x$};
    \node[mynode, right =1cm of x1] (y1) {$v_y$};
    \node[mynode, below right = 1cm of y1] (j1) {$v_j$};
    \node[mynode, below right = 1cm of x1] (z1) {$v_z$};
    
    \draw (x1) -- (y1);
    \draw (i1) to[out=-45,in=-135] (j1);
    \draw (i1) -- (x1);
    \draw[dashed] (j1) -- (y1);
    \draw[dashed] (x1) -- (z1);
    \draw (z1) -- (j1);
    
    \end{tikzpicture}
    }
    \caption{Case 3.}
    \label{fig: MM prefix 3}
    \end{subfigure}
&
    \begin{subfigure}{.32\textwidth}
    \centering
    \resizebox{0.8\textwidth}{!}{
    \tikzset{
    myptr/.style={-{Stealth[scale=1.5]}},
    }
    
    \definecolor{b}{rgb}{0.0, 0, 1}
    \definecolor{gr}{rgb}{0.7, 0.7, 0.7}
    
    \begin{tikzpicture}[
        mynode/.style={circle, draw=black, minimum size=20pt, inner sep=0pt}
    ]

    \node[mynode] (i) {$v_i$};
    \node[mynode, above right = 1cm of i] (x) {$v_x$};
    \node[mynode, right =1cm of x] (y) {$v_y$};
    \node[mynode, below right = 1cm of y] (j) {$v_j$};
    \node[right = 0.2cm of x] (sup) {};
    \node[mynode, below left = 1cm of y] (z) {$v_z$};
    \node[right = 0.4cm of x] (sup) {};
    
    \draw (x) -- (y);
    \draw[dashed] (i) to[out=-45,in=-135] (j);
    \draw (i) -- (x);
    \draw (j) -- (y);
    \draw[dashed] (y) -- (z);
    \draw (i) -- (z);

    \node[below=2.5cm of sup] (m1) {};
    \node[below=0.5cm of m1] (m2) {};
    \draw[->]        (m1)   -- (m2) node[pos=0.5, right=0.2cm] {\small FLIP};
    
    \node[mynode, below = 3.5cm of i] (i1) {$v_i$};
    \node[mynode, above right = 1cm of i1] (x1) {$v_x$};
    \node[mynode, right =1cm of x1] (y1) {$v_y$};
    \node[mynode, below right = 1cm of y1] (j1) {$v_j$};
    \node[mynode, below left = 1cm of y1] (z1) {$v_z$};
    
    \draw (x1) -- (y1);
    \draw (i1) to[out=-45,in=-135] (j1);
    \draw (i1) -- (x1);
    \draw[dashed] (j1) -- (y1);
    \draw (y1) -- (z1);
    \draw[dashed] (i1) -- (z1);
    
    \end{tikzpicture}
    }
    \caption{Case 4.}
    \label{fig: MM prefix 4}
    \end{subfigure}
\end{tabular}
\caption{Illustration of the cases in the proof of Claim 3.
All the illustrating figures maintain the convention that 
solid lines represent edges existing in the graph and dashed lines represent edges not in the graph.
}
\label{fig: MM prefix cases}
\end{figure}

\noindent
Combining \Cref{cl: Gould for MM} and \Cref{cl: MM structure} gives a more specific Prefix Lemma for \matchlong.

\begin{lemma}
{\bf (Inverted Prefix Lemma for \match)} 
\label{lem: MM prefix}
If a sequence $d$ has a realization $G = (V, E)$ where $V = \{v_1, v_2,\ldots, v_n\}$ such that $\deg_G(v_i) = d_i$ for every $i \in [1,n]$,
with a matching $M$ of size $\nu$, then $d$ has a realization $G' = (V, E')$ with the 
matching $M' = \{(v_i, v_{2\nu-i+1}) \mid i \in [1,\nu]$ and the same degrees $\deg_{G'}(v_i) = d_i$ for every $i \in [1,n]$.
\end{lemma}
\subsection{Reduction to a Bipartite Sequence Pair}

A bipartite graph $\hat{G} = (V, W, \hat{E})$, where $V = \{v_1, v_2, \ldots, v_n\}$ and $W = \{w_1, w_2, \ldots, w_n\}$, is a \emph{$\nu$-matched} realization for the sequence pair $(d, d)$ if it satisfies the following properties.
\begin{compactenum}[(M1)]
\item $\hat{G}$ realizes the sequence pair $(d,d)$.
\item $(v_i, w_i) \notin \hat{E}$ for every $i \in [1,n]$.
\item $\hat{M} = \{(v_i, w_{2\nu-i+1}) \mid i \in [1,2\nu] \subseteq \hat{E}$.
\end{compactenum}

First, we describe an algorithm that given a $\nu$-matched realization $\hat{G}$ of $(d,d)$, produces a graph $G$ realizing $d$ with a matching of size $\nu$.

\begin{enumerate}
\item
\label{step: 1}
Compute a half-integral solution.
\begin{compactenum}
\item
For all $i, j \in [1,n]$, let 
\begin{align*}
y_{ij} & 
= \begin{cases}
    1, ~ \{v_i, w_j\} \in \hat{E}, \\
    0, ~ \text{otherwise,}
\end{cases}
&
\text{and}
&& 
\omega(i,j) = \frac{1}{2}(y_{ij} + y_{ji}).
\end{align*}

\item 
Define a weighted graph $G^\omega = (V^\omega, E^\omega,\omega)$ with vertex set $V^\omega = [1,n]$ and 
an edge $e=(i, j)$ of weight $\omega(e)=$ $\omega(i,j)$ for every $i,j \in V^\omega$. 
Clearly, $w$ is half-intergral.
\item Let $M^\omega= \{(i, 2\nu-i+1) \mid ~i \in [1,\nu]\}$ be the inverted matching in $G^\omega$. 
According to (M3)
\begin{align}
\label{eq: matching edges}
    \omega(e) = 1, ~\text{for every $e \in M^\omega$}.
\end{align}
\item
Define the \emph{weighted degree} of a vertex $i \in V^\omega$ to be $d^\omega(i) = \sum_{j\in V^\omega} \omega(i, j)$. Note that $G^\omega$ realizes $d$ in the \emph{weighted} sense, namely,
\\
\hbox{\hskip 20pt}
$
    d^\omega(i) = \sum_{j\in V^\omega} \omega(i, j) = \frac{1}{2} \left(\sum_{j=1}^n y_{ij} + \sum_{j=1}^n y_{ji}\right) = d_i, ~ \mbox{for any}~ i \in V^\omega.~~~~~~~\mbox{}
$
\end{compactenum}

\item
\label{step: 2} 
Preparing for discarding non-integral weights while keeping the degrees unchanged.
\\
Construct a graph $\Ghalf = (\Vhalf,\Ehalf)$ obtained by removing from $G^\omega$ the edges of integral weight and keeping only those of weight $1/2$.
Formally, $\Vhalf = V^\omega$ and $\Ehalf = \{e \in E^\omega \mid \omega(e) = 1/2 \}$.

\medskip\noindent
In later stages of the construction, whenever $G^\omega$ is modified (by changing the weight of some edge $e$ from $1/2$ to 0 or 1), $\Ghalf$ is modified accordingly (by removing the edge $e$).

\item
\label{step: 3}
Partition into cycles.
\begin{compactenum}
\item 
Partition the edge set of $\Ghalf$ into disjoint (not necessarily simple) cycles each covering an entire connected component of the graph. Since $\Ghalf$ is an even graph by \Cref{obs: G1/2 is even graph for MM}, this can be done in polynomial number of steps.
\item
Let $\calC$ be a set of cycles in the partition.
\item
Let $\calC^{even} \gets \{C\in\calC \mid C ~\mbox{is of even length}\}$,
~~~~~
$\calC^{odd} \gets \calC \setminus \calC^{even}$.
\end{compactenum}

\item 
\label{step: 4}
Eliminate even cycles.

\noindent
For every cycle $C \in \calC^{even}$, do the following.
\begin{compactenum}
\item 
Traverse $C$ starting from an arbitrary vertex $x \in C$ and continuing along the cycle until returning to $x$. Denote  the resulting sequence of edges by $E(C)=(e_1, e_2, \ldots, e_\ell)$. 
\item
Increase (resp., decrease) edge weights on even (resp., odd) positions in the sequence $E(C)$ by $1/2$.
That is, for every $i \in [1,\ell]$,
$\omega(e_i)$ is set to 1 if $i$ is even and 0 otherwise. 
(This does not change the weighted degrees in $G^\omega$ and does not affect other cycles in $\calC$ or edges in $M^\omega$.)
\end{compactenum}

\item 
\label{step: 5}
Eliminate odd cycles.

\noindent
Arrange the cycles in $\calC^{odd}$ in pairs 
(recall that by \Cref{obs: even number of odd cycles for MM} their number is even).
For every pair of cycles $(C,C')$ choose vertices $x \in C$, $y \in C'$ according to Obs. \ref{obs: right x and y for MM}. Then do the following.
\begin{compactenum}
\item 
Starting from $x \in C$, traverse $C$ via edges 
$E(C)=(e_1, e_2, \ldots, e_\ell)$.
\\
Starting from vertex $y \in C'$, traverse $C'$ via edges 
$E(C')=(e'_1, e'_2, \ldots, e'_k)$.
\item 
Let $\xi = \omega(x,y)$ 
(by \Cref{obs: right x and y for MM}, this weight must be either $0$ or $1$).
\item Set $\omega(x, y) \gets 1-\xi$.
\item For every $i \in [1,\ell]$ and $j \in [1,k]$, 
      modify the edge weights in the cycles $C$ and $C'$ as follows:
\begin{align*}
\omega(e_i) & \gets 
    \begin{cases}
        1-\xi, ~ \text{$i$ is even}, \\
        \xi, ~ \text{$i$ is odd},
    \end{cases} 
&
\omega(e'_j) & \gets 
    \begin{cases}
        1-\xi, ~ \text{$j$ is even}, \\
        \xi, ~ \text{$j$ is odd}.
    \end{cases} 
\end{align*}
\end{compactenum}
By \Cref{obs: G^omega realizes d for MM}, at this stage each edge in $G^\omega$ 
has weight 1 or 0, $G^\omega$ realizes $d$, and Eq.~\eqref{eq: matching edges} holds.

\item 
\label{step: 6}
Generate the output $G$.

\noindent
Transform $G^\omega$ into a simple graph $G$, with an edge $(i,j)$ whenever $\omega(i,j) = 1$ in $G^\omega$. Note that $M = M^\omega$ forms a matching in $G$ by Eq.~\eqref{eq: matching edges}.
\end{enumerate}


The algorithm is justified by the following Observations.

\begin{observation} 
\label{obs: G1/2 is even graph for MM}
The graph $G^{1/2}$ is even
(namely, all its vertex degrees are even).
\end{observation}

A connected even graph has an Euler cycle, implying the following.

\begin{observation} 
\label{obs: even number of odd cycles for MM}
The number of cycles in $\calC^{odd}$ is even.
\end{observation}
\begin{proof}
Observe that
$\sum_{e \in E^\omega} \omega(e) ~=~ \frac{1}{2} \sum_{i=1}^n d_i = m,$
where $m$ is the number of edges, which is an integer.
Therefore, the number of edges with weight $1/2$ must be even. 
Since the cycles in $\calC$ cover all of $\Ehalf$ and are disjoint, the observation follows.
\end{proof}

\begin{observation}
\label{obs: right x and y for MM}
For any pair of disjoint odd cycles $C, C' \in \calC^{odd}$, there exist $x \in C$ and $y \in C'$, such that $(x,y) \notin M^\omega$ and $\omega(x, y) \neq 1/2$.
\end{observation}
\begin{proof}
Let $x \in C$, $y' \in C'$ be arbitrary vertices. 
If $(x, y') \in M^\omega$, let $y \neq y'$ be any other vertex in $C'$. Otherwise, let $y = y'$.
Since $M^\omega$ is a matching, it follows that $(x, y) \notin M^\omega$. 

Assume, towards contradiction, that $\omega(x,y)=1/2$. Then the edge $(x,y)$ belongs to $G^{1/2}$, implying that $x$ and $y$ belong to the same connected component in $G^{1/2}$. But this contradicts the fact that $x\in C$ and $y\in C'$ where $C$ and $C'$ cover different connected components of $\Ghalf$. 
\end{proof}

\begin{observation} 
\label{obs: G^omega realizes d for MM}
At the end of Step \ref{step: 5}, each edge in 
$G^\omega$ has weight 1 or 0, $G^\omega$ 
realizes $d$ and Eq, \eqref{eq: matching edges} holds.
\end{observation}
\begin{proof}
After Step \ref{step: 5}, $G^\omega$ contains no more edges of weight $1/2$. Moreover, neither of the modifications performed in Step \ref{step: 5} change the weighted degrees in $G^\omega$, affect other cycles in $\calC$ or edges in $M^\omega$. The observation follows.
\end{proof}

The reduction correctness follows from the next lemma.
\begin{lemma}
\label{lem: MM equivalence}
There is a $\nu$-matched realization of $(d,d)$ if and only if $\match(d) \ge \nu$.
\end{lemma}
\begin{proof}
$(\Rightarrow)$ If there $\nu$-matched realization of $(d,d)$, then the aforementioned algorithm produces a realization of $d$ with matching of size $\nu$.
    
$(\Leftarrow)$ Given a realization $G = (V, E)$ of $d$ with vertices $V = [1,n]$ and 
a matching $M \subseteq E$ of size $\nu$, there is a realization $G' = (V, E')$ with 
a matching $M' = \{(v_i, v_{2\nu-i+1}) \mid i \in [1,\nu] \}$ by the Inverted Prefix Lemma (\Cref{lem: MM prefix}). 
Consider the following graph $\hat{G} = (\hat{V}, \hat{W}, \hat{E})$ with $\hat{V} = \{v_1, \ldots, v_n\}$, $\hat{W} = \{w_1, \ldots, w_n\}$ and $(v_i, w_j) \in \hat{E}$ if and only if $(i, j) \in E'$. 
Clearly, $\hat{G}$ realizes $(d,d)$ and $(v_i,w_i) \notin \hat{E}$ for every $i \in [1,n]$, 
so it satisfies (M1) and (M2). Moreover, the set $\hat{M} = \{(v_i, w_j) \mid (i,j) \in M'\}$ 
is a matching of size $2 \nu$ in $\hat{G}$ satisfying $(M3)$. Indeed, since $M'$ is a matching, 
the edges in $\hat{M}$ form a matching too and each edge $(i,j) \in M'$ leads to two edges 
$(v_i, w_j), (v_j, w_i) \in \hat{M}$. The Lemma follows.
\end{proof}

\subsection{Reduction to Flow}

According to the previous section (in particular, \Cref{lem: MM equivalence}, given a degree sequence $d$ all we need is to find a bipartite $\nu$-matched realization of $(d,d)$ for maximum possible $\nu$. Next we describe how to construct a flow graph $G_{d,\nu}$ for any $\nu$, such that $\nu$-prefix realization of $(d,d)$ exists if and only if the maximum flow in $G_{d,\nu}$ is $\sum_{i=1}^n d_i - 2\nu$. Moreover, given an integer flow of this value we describe how to construct a desired realization. To find a solution to MM-DR one simply needs to iterate over $\nu \in [1,n]$ and find the maximum one with the aforementioned maximum flow.

First we describe construction of $G_{d, \nu}$.
The source $s$ is connected to nodes $x_i \in X$ (for $i \in [1,n]$) that correspond to the vertices of $V$. 
Similarly, the node set $Y$ contains nodes $y_j$ (for $j \in [1,n]$), 
corresponding to the vertices in $W$, and these are connected to the sink $t$.
The nodes are organized into the following sets:
\begin{compactitem}
\item 
Candidates for the matching: $X_M =\{x_i \mid i \in [1,2\nu]$ and $Y_M = \{y_j \mid j \in [1,2\nu]$,
\item 
Rest of the nodes: 
$X_R = \{x_i \mid i \in [2\nu+1,n]\}$ and $Y_R = \{y_j \mid j \in [2\nu+1,n]\}$,
\item
Source and sink: \{s,t\}.
\end{compactitem}
The edges are capacitated and directed from left to right, i.e., an edge $(\alpha, \beta, \gamma)$ leads from the node $\alpha$ to the node $\beta$ and can carry up to $\gamma$ units of flow. The source $s$ has edges leading to every node $x_i$, with capacity $d_i-1$ for $x_i \in X_M$ and $d_i$ for $x_i \in X_R$. Similarly, there are edges leading from every node $y_j$ to the sink $t$, with capacity $d_j-1$ for $y_j \in Y_M$ and $d_j$ for $y_j \in Y_R$. This enforces the degree constraint for the vertices $v_i$ and $w_i$ in $G$. Finally, there's a unit capacity edge $(x_i,y_j,1)$ for every $i, j \in [1,n]$, such that $i \neq j$ and $j \neq 2\nu - i + 1$. 
Overall, the edge set of the flow graph is defined as follows:
\begin{align*}
\tilde{E} ~=~
& \set{(s, x_i, d_i-1) \mid i \in [1,2\nu]} \cup \set{(s, x_i, d_i) \mid i \in [2\nu+1,n]} \\
& \cup \set{(x_i, y_j, 1) \mid i, j \in [1,n], ~ i\neq j, j \neq 2\nu - i + 1} \\
& \cup \set{(y_j, t, d_j-1) \mid j \in [1,2\nu]} 
  \cup \set{(y_j, t, d_j) \mid j \in [2\nu + 1,n]}
~.
\end{align*}

\tikzset{
  roundNode/.style={
    circle, draw, minimum size=1cm, align=center, font=\Large
  },
  dots/.style={
    draw=none, fill=none, yshift=0.1cm
  },
  edge/.style={
    font=\large
  }
}

The construction is justified by the following lemma.

\begin{lemma}
\label{lem: matching to flow}
There exists a $\nu$-matched bipartite graph $\hat{G} = (V, W, \hat{E})$ that realizes the $(d,d)$ 
if and only if the value of the maximum $s-t$ flow in $G_{d,\gamma}$ is equal to $\sum_{i=1}^n d_i-2\nu$.
\end{lemma}
\begin{proof}
Suppose the value of the maximum $s-t$ flow in $G_{d, 2\nu}$ is $\sum_{i=1}^n d_i-2\nu$. 
Define the bipartite graph $\hat{G} = (V, W, \hat{E})$ as follows:
\begin{compactenum}
\item For each node $x_i \in X$, create a corresponding node $v_i \in V$.
\item For each node $y_j \in Y$, create a corresponding node $w_j \in W$.
\item For every $1\le i, j\le n$, if $\FLOW(x_i,y_j)=1$, then create an edge $(v_i, w_j)$ and add it to $\hat{E}$.
\item  For every $1\le i\le 2\nu$ create an edge $(v_i, w_{2\nu-i+1})$ and add it to $\hat{E}$.
\end{compactenum}

It remains to verify that $\hat{G}$ correctly realizes $(d,d)$ and is $\nu$-matched.
The degree of $v_i$ in $\hat{G}$ equals to the number of edges $(x_i,y_j)$ that carry flow if $2\nu \le i \le n$ and one more than the corresponding flow if $1 \le i \le 2 \nu$. Since $x_i$ conserves flow, and since all the edge $(s,x_i)$ are saturated, $\deg(v_i) = d_i$. Similarly, $\deg(w_j) = d_j$ for every $j \in [1,n]$, so $\hat{G}$ satisfies (M1). Since $G_{d,\gamma}$ does not have edges $(x_i,y_i)$ for any $i \in [1,n]$, it follows that $\hat{G}$ does not have edges $(v_i, w_i)$ satisfying (M2). Finally, the matching $\hat{M} = \{(v_i, w_{2\nu-i+1)} \mid 1 \le i \le 2\nu\}$ was added in step 4, so $\hat{G}$ satisfies (M3).

Conversely, suppose $\hat{G} = (V, W, \hat{E})$, where $V = \{v_1, \ldots, v_n\}$ and $W = \{w_1, \ldots, w_n\}$, 
is a $2\nu$-matched bipartite graph realizing $(d, d)$ with the matching $\hat{M} = \{(v_i, w_{2\nu-i+1}) \mid 1 \le i \le 2\nu\}$.
We need to define a flow function on $G_{d,\nu}$ and prove that its value 
equals $\sum_{i=1}^n d_i - 2\nu$. 
Obviously, maximum flow cannot exceed this sum, thus it is enough to construct flow with aforementioned value.

\begin{enumerate}
\item \textbf{Saturate Edges from Source and to Sink}:
\begin{align*}
\FLOW(s,x_i) & \gets 
\begin{cases}
d_i - 1 & i \in [1,2\nu], \\
d_i     & i \in [2\nu+1,n].
\end{cases}
& 
\FLOW(y_j,t) & \gets 
\begin{cases}
d_j - 1 & j \in [1,2\nu], \\
d_j     & j \in [2\nu+1,n].
\end{cases}
\end{align*}
\item 
\textbf{Flow Through Edges}:
\\
For every $i \in [1,n]$ and $j \in [1,n]$ such that $(v_i, w_j) \in \hat{E} \setminus \hat{M}$, set $\FLOW(x_i, y_j) \leftarrow 1$,
\item 
\textbf{Completing the flow function}: 
\\
Set the flows of all other edges to $0$.
\end{enumerate}

We need to argue that the defined flow is legal. 
Note that by construction, the total flow out of the source $s$ is $\sum_{i=1}^n d_i - 2\nu$ 
and the total flow into the sink $t$ is $\sum_{j=1}^n d_j - 2\nu$. It remains to verify that 
for every vertex except $s$ and $t$, the incoming and outgoing flows are equal.
\begin{itemize}
\item \textbf{Node $x_i$}: 
By construction, the incoming flow at $x_i$ is $\FLOW(s, x_i) = d_i-1$ if $x_i \in X_M$ and $\FLOW(s, x_i) = d_i$ if $x_i \in X_R$. 
As for the outgoing flow, recall that 
there are exactly $d_i-1$ indices $j$ for which $(v_i,w_j)\in \hat{E} \setminus \hat{M}$ if $i \in [1,2\nu]$ and $d_i$ 
such indices if $i \in [2\nu+1,n]$. 
Hence, $x_i$ conserves flow.
\item \textbf{Node $y_j$}:
By construction, the outgoing flow at $y_j$ is $\FLOW(y_j, t) = d_j-1$ if $x_j \in X_M$ and $\FLOW(y_j, t) = d_j$ if $y_j \in X_R$. 
As for the incoming flow, recall that 
there are exactly $d_j-1$ indices $i$ for which $(v_i,w_j)\in \hat{E} \setminus \hat{M}$ if $j \in [1,2\nu]$ 
and $d_j$ such indices if $j \in [2\nu+1,n]$. 
Hence, $y_j$ conserves flow.
\end{itemize}
The claim follows.
\end{proof}

The running time analysis of the realization algorithm for \matchDR 
is similar to the one for \domsetDR.
We conclude the section with the following result.

\begin{theorem}
\label{thm:real-MM-fast}
There exists an $O(n^3)$ time algorithm for constructing a realization of a 
given graphic sequence $d$ that also has a matching of maximum size 
(among all possible realizations of $d$).
\end{theorem}

\clearpage

\bibliographystyle{abbrv}
\bibliography{realizations}

\end{document}